\documentclass[11pt,letterpaper]{article}
\usepackage[margin=1in]{geometry}
\usepackage{graphicx}
\usepackage{multirow}
\usepackage{xcolor,xspace}
\usepackage{float}

\usepackage{subfigure}
\usepackage{tabularx}
\usepackage{mdframed}
\usepackage{graphicx}
\usepackage{float}
\usepackage{subfigure}
\usepackage{comment}
\usepackage{subfigure}

\usepackage[sort]{cite}
\usepackage{hyperref}

\usepackage{colortbl}

\usepackage{paralist}
\usepackage{enumitem}
\usepackage{booktabs}
\usepackage[override]{cmtt}
\usepackage{color,xcolor}
\usepackage{graphicx,color}

\usepackage{tikz}
\usetikzlibrary{arrows.meta,positioning,calc}

\usepackage{boxedminipage}
\usepackage{url}
\usepackage{graphicx}
\usepackage{color}
\usepackage{xspace}
\usepackage{amsmath,amsthm,amstext,amssymb,amsfonts,latexsym}
\usepackage[capitalize,noabbrev]{cleveref}

\usepackage{algorithm}
\usepackage{algorithmic}

\usepackage{tcolorbox}

\definecolor{darkred}{rgb}{0.5, 0, 0}
\definecolor{darkgreen}{rgb}{0, 0.5, 0}
\definecolor{darkblue}{rgb}{0,0,0.5}

\usepackage{listings}
\usepackage{xcolor}

\lstdefinestyle{sqlstyle}{
  language=SQL,
  basicstyle=\ttfamily\small,
  keywordstyle=\color{blue},
  commentstyle=\color{gray},
  stringstyle=\color{black},
  numbers=left,
  numberstyle=\tiny\color{gray},
  stepnumber=1,
  numbersep=10pt,
  backgroundcolor=\color{white},
  showspaces=false,
  showstringspaces=false,
  showtabs=false,
  frame=single,
  tabsize=2,
  captionpos=b,
  breaklines=true,
  breakatwhitespace=false
}

\makeatletter
\newcommand{\algline}[1]{%
  \edef\@currentlabel{\arabic{ALC@line}}%
  \label{#1}%
}
\makeatother

\newcounter{note}[section]
\newcommand{\hubert}[1]{\refstepcounter{note}$\ll${\sf Hubert's
        Comment~\thenote:} {\sf \textcolor{blue}{#1}}$\gg$\marginpar{\tiny\bf HC~\thenote}}

\newcommand{\yuetodo}[1]{{\large\color{green}[Yue todo: #1]}}

\definecolor{yue}{rgb}{0.7, 0, 0}

\renewcommand{\yuetodo}[1]{}

\newcommand{\mcal}[1]{\ensuremath{\mathcal {#1}}}

\newcommand{\algP}{\ensuremath{\mcal{P}}\xspace}

\newcommand{\algA}{{\ensuremath{\mcal{A}}}\xspace}

\definecolor{darkgreen}{rgb}{0,0.5,0}
\definecolor{lightblue}{RGB}{0,176,240}
\definecolor{darkblue}{RGB}{0,112,192}
\definecolor{lightpurple}{RGB}{124, 66, 168}
\definecolor{grey}{RGB}{139, 137, 137}
\definecolor{maroon}{RGB}{178, 34, 34}
\definecolor{green}{RGB}{34, 139, 34}
\definecolor{types}{RGB}{72, 61, 139}
\definecolor{gold}{rgb}{0.8, 0.33, 0.0}

\definecolor{darkgray}{gray}{0.3}

\definecolor{darkred}{rgb}{0.5, 0, 0}
\definecolor{darkgreen}{rgb}{0, 0.5, 0}
\definecolor{darkblue}{rgb}{0,0,0.5}

\newcommand\markx[2]{}

\newcommand{\R}{\mathbb{R}}

\newcommand{\N}{\mathbb{N}}

\newcommand{\Z}{\mathbb{Z}}

\newcommand{\ignore}[1]{}

\newcounter{task}

\newtheorem{theorem}{Theorem}[section]

\newtheorem{corollary}[theorem]{Corollary}
\newtheorem{fact}[theorem]{Fact}
\newtheorem{lemma}[theorem]{Lemma}
\newtheorem{proposition}[theorem]{Proposition}
{
\theoremstyle{definition}
\newtheorem{definition}[theorem]{Definition}
\newtheorem{remark}[theorem]{Remark}
\newtheorem{assumption}[theorem]{Assumption}
}

\newcounter{cnt:challenge}

\newcommand{\renyi}{R\'{e}nyi\xspace}
\newcommand{\D}{\mathsf{D}}

\newcommand{\NA}{{\textsc{NA}}}

\newcommand{\supp}{\ensuremath{\mathsf{supp}}}

\newcommand{\f}{\mathfrak{f}}

\newcommand{\algM}{\ensuremath{\mcal{M}}\xspace}
\newcommand{\algN}{\ensuremath{\mcal{N}}\xspace}
\newcommand{\halt}{\mathsf{halt}}

\newcommand{\Pow}{\mathsf{Pow}}
\newcommand{\view}{\mathsf{view}}
\newcommand{\pair}{\mathsf{pair}}
\newcommand{\npp}{\mathsf{npp}}

\newcommand{\hpresum}{\mathsf{HybridStreamSum}}

\newcommand{\presum}{\mathsf{StreamSum}}
\newcommand{\buf}{\mathsf{Buf}}

\newcommand{\bin}{\mathsf{Bin}}
\newcommand{\ind}{\mathsf{Index}}
\newcommand{\randbin}{\ensuremath{\mathsf{RandBin}}\xspace}

\newcommand{\DP}{\mathsf{DP}}

\newcommand{\simB}{\sim_{\mathsf{B}}}
\newcommand{\simE}{\sim_{\mathsf{E}}}

\newcommand{\PStreamSum}{\ensuremath{\mathsf{PrivStreamSum}}\xspace}

\title{Continual Learning With Participation Privacy: An Auditable Buffering-Aggregation Recipe}

\author{T-H. Hubert Chan$^{1}$ \qquad Elaine Shi$^{2}$ \qquad Mengshi Zhao$^{1}$ \qquad Mingxun Zhou$^{3}$ \\[0.6em] \small $^{1}$The University of Hong Kong, Hong Kong, China \\ \small $^{2}$Carnegie Mellon University, Pittsburgh, USA \\ \small $^{3}$The Hong Kong University of Science and Technology, Hong Kong, China }

\date{}

\begin{document}
\begin{titlepage}
\maketitle

\begin{abstract}

Modern federated and streaming learning systems often release intermediate
models, so privacy must hold for the full trajectory under adaptive
interaction. Motivated by participation privacy, we study single-edit
neighboring user streams, where one insertion/deletion shifts all subsequent
updates and defeats standard Hamming-neighbor continual-release analyses. We
give an auditable modular recipe. A randomized buffering wrapper emits bins of
size $[U,2U]$, reducing single-edit streams to a Hamming-style per-bin update
stream with explicit backlog/delay guarantees, where $U$ is calibrated by the
privacy parameters $(\varepsilon,\delta)$. We then prove a certification
theorem identifying when a non-adaptive Hamming-neighbor DP proof for a
continual primitive lifts to adaptive inputs: the primitive must use fresh
per-round randomness and have a stable one-round privacy profile under common
adaptive context. Together, these ingredients yield trajectory-level
$(\varepsilon,\delta)$-DP for single-edit streams using standard primitives
(e.g., tree prefix sums), with an explicit privacy--latency link via $U$.

\end{abstract}

\thispagestyle{empty}

\end{titlepage}

\pagestyle{plain}

\section{Introduction}
\label{sec:intro}

Modern federated and streaming learning systems release intermediate model snapshots throughout training, so an adversary may observe the entire trajectory $(w^{(t)})_{t\ge 0}$ rather than only a final model.
A canonical instance is streaming stochastic gradient descent (SGD): at step~$t$, a user contributes a private loss function $f_t$ and the learner updates using a gradient step based on the current model~\cite{DBLP:conf/nips/ThakurtaS13,DBLP:conf/icml/KairouzM00TX21}.
Beyond protecting the data within each $f_t$, some deployments also require \emph{participation privacy}: a user may wish to ensure \emph{plausible deniability} of whether they participated at all.
This motivates \emph{single-edit} (edit-distance) neighboring streams, where two user-update streams are neighbors if one can be obtained from the other by inserting or deleting a single user event.

Continual release also makes the learning process intrinsically interactive: the released model $w^{(t)}$ influences what happens next, so the update stream can be adaptively generated in response to prior outputs.
While classical continual-release primitives such as tree-based private prefix sums~\cite{DBLP:conf/stoc/DworkNPR10,DBLP:journals/tissec/ChanSS11} were originally analyzed under statically chosen (Hamming-neighbor) streams, recent work formalizes differential privacy against such adaptive interaction~\cite{DBLP:conf/icml/0004RSS23}.
Our goal is to make this adaptive viewpoint compatible with the \emph{single-edit} participation model above.

The difficulty is that Hamming-style neighboring streams are not the right abstraction for participation privacy on a stream.
Most continual DP analyses change the value at a single time index $t_0$ while keeping all other positions aligned~\cite{DBLP:conf/icalp/Dwork06}.
In contrast, a single insertion/deletion shifts the alignment of all subsequent positions, so an edit-neighbor pair can be far from Hamming-neighboring.
This also undermines deterministic batching: partitioning the stream into fixed contiguous blocks can cause one edit to shift many later batch boundaries, so standard Hamming-neighbor continual DP guarantees do not directly transfer.  

To address this mismatch, we use a simple modular recipe: randomized buffering and an adaptive-safety certification for downstream continual primitives.
We apply a \emph{randomized buffering} wrapper~\cite{DBLP:journals/jacm/ChanCMS22,DBLP:conf/eurocrypt/ZhouSCM23} that introduces random delay and releases user updates in \emph{bins} of controlled size (e.g., in $[U,2U]$); here $U$ is a \emph{privacy-implied systems cost}, calibrated by the target $(\varepsilon,\delta)$ and inducing delay.
We then invoke standard continual DP primitives (e.g., tree/prefix-sum) to
privately aggregate the resulting vector updates, and prove a
\emph{certification theorem} showing when a non-adaptive Hamming-neighbor
privacy proof remains valid under feedback.
A single edit can alter at most two adjacent emitted bins; adaptive group
privacy therefore extends the certified distance-one Hamming guarantee to the
resulting downstream prefix-sum inputs.  
The certification requires fresh
per-round randomness together with a stable-context condition: common adaptive
updates may change the released statistic, but they must not change the
privacy cost of the unique Hamming discrepancy.\footnote{The pre-conference version~\cite{DBLP:conf/icml/ChanSZZ26} stated the certification condition using independent decomposability alone.
However, we use the slightly stronger stable-context condition, formalized in
Appendix~\ref{sec:ind_decomp} as quantitative decomposability.  This
refinement makes the statement technically precise and is satisfied by the
certified cores of the continual-release primitives used in our pipeline.}

In the adaptive privacy game, the two candidate functions or bins are
converted to updates using the same currently observed model within each
branch; the two realized update streams need not remain Hamming-neighboring
after their model trajectories diverge.

\noindent\textbf{Contributions.}
\begin{itemize}\setlength{\itemsep}{0.15em}\setlength{\topsep}{0.15em}
  \item \textbf{Edit-neighbor continual learning under feedback.}
  We formalize participation privacy via single-edit neighboring user-update streams in a continual learning setting with intermediate releases.

  \item \textbf{Modular pipeline with an explicit privacy to latency link.}
  We combine randomized buffering with standard continual DP primitives to obtain trajectory-level privacy for edit-neighbor streams; the target $(\varepsilon,\delta)$ determines the buffering level $U$, making induced delay explicit.

\item \textbf{Certification theorem for adaptive safety of continual DP primitives.}
We give a checkable certification condition under which a non-adaptive
Hamming-neighbor privacy analysis carries over to the feedback/adaptive
setting.  The condition combines fresh per-round randomness with a
stable-context requirement, and certifies safe reuse of canonical primitives
such as tree/prefix-sum mechanisms.

\end{itemize}

\noindent\textbf{Results and organization.}
Our end-to-end guarantee is trajectory-level $(\varepsilon,\delta)$-DP for \emph{single-edit} neighboring participation streams under feedback, obtained by composing randomized buffering with a continual-release DP aggregator (Theorem~\ref{thm:main_pipeline_overview}); the privacy target explicitly induces a buffering level $U(\varepsilon,\delta)$ and hence delay/staleness.

A second ingredient is a reusable \emph{certification} theorem showing that,
for continual primitives satisfying our fresh-randomness and stable-context
conditions, a standard non-adaptive privacy analysis lifts to the
adaptive/feedback setting (Theorem~\ref{thm:cert}).
Section~\ref{sec:setting} formalizes the interactive model, edit adjacency, and metrics; Section~\ref{sec:overview} gives a high-level view of the pipeline and states the main theorem; Section~\ref{sec:ingredients} develops the two auditable ingredients (randomized buffering and certification) and combines them via modular composition; Section~\ref{sec:instantiation} instantiates the recipe with continual DP-SGD (tree/prefix-sum aggregation).
Technical proofs and additional instantiations are deferred to the appendix.
Experiments appear in~\cite{DBLP:conf/icml/ChanSZZ26}.

\subsection{Related Work}

\textbf{Continual-Release DP.}
Canonical continual-release primitives for streaming aggregates include tree-based private prefix sums~\cite{DBLP:conf/stoc/DworkNPR10,DBLP:journals/tissec/ChanSS11} and closely related lower-triangular (matrix-style) mechanisms~\cite{DBLP:conf/nips/DenisovMRST22}.
These mechanisms are typically analyzed for \emph{Hamming-neighboring} update streams, where only one time index differs.
In contrast, our participation model uses \emph{single-edit} (insertion/deletion) neighboring streams, where alignment shifts after the edit; this gap is mild for static datasets~\cite{DBLP:conf/nips/BirrellEBP24} but is fundamental for streams.

Sliding-window variants of private stream release have also been studied for
window aggregate queries and graph streams~\cite{DBLP:conf/edbt/CaoXGLBT13,DBLP:conf/aistats/UpadhyayUA21}.
These works target recent-window utility, whereas our setting focuses on
single-edit participation privacy, where one insertion/deletion shifts all
subsequent stream positions and requires an edit-to-Hamming shielding interface.

\textbf{Adaptive Streams and Feedback.}
Continual learning is inherently interactive: released outputs can influence future updates.
Recent work formalizes differential privacy against such adaptive interaction (e.g., left-or-right style games and verification viewpoints)~\cite{DBLP:conf/nips/DenisovMRST22,DBLP:conf/icml/0004RSS23}.
Our \emph{certification} result is complementary: it gives a checkable
structural condition under which a standard non-adaptive Hamming-DP analysis
of a continual primitive lifts directly to the adaptive setting.  The
condition is stronger than fresh prefix-causal randomness alone: it also
requires that common adaptive context does not change the one-round privacy
cost of the unique Hamming discrepancy.

\textbf{Privacy-Induced Delayed Reactions.}
\cite{DBLP:conf/colt/Cohen0NSS24} give a conceptually related delayed-reaction
phenomenon for continual observation and online threshold queries. Their
Appendix~B delays threshold-type outputs on a directly observed online bit
stream. Our delay has a different technical role: \randbin waits for future
arrivals to form bins in $[U,2U]$ so that a single insertion/deletion, which
appears only in the neighboring-stream analysis, becomes a bounded
Hamming-style perturbation while the full transcript, including timing and
$\bot$ outputs, remains protected.

\textbf{Randomized Buffering, RSC, and Obliviousness Lineage.}
Randomized buffering has appeared as infrastructure in oblivious data
structures and related privacy modularity work~\cite{DBLP:journals/jacm/ChanCMS22,DBLP:conf/eurocrypt/ZhouSCM23}.
A closely related random-partition idea also appears in the
Reorder-Slice-Compute (RSC) paradigm of~\cite{DBLP:conf/stoc/Cohen0NSS23}, where random
slicing is used to synchronize neighboring executions and avoid a one-edit
domino effect. Our setting differs in that the continual streaming transcript
exposes release timing itself: the adversary observes whether each step emits a
real block or the dummy symbol $\bot$. Thus, in our use, \randbin is an
interactive interface that converts single-edit participation streams into a
sparse, Hamming-style per-bin update stream while also providing explicit
backlog/delay guarantees calibrated by $(\varepsilon,\delta)$.

\ignore{
\textbf{Randomized Buffering and Obliviousness Lineage.}
Randomized buffering has appeared as infrastructure in oblivious data structures and related privacy modularity work~\cite{DBLP:journals/jacm/ChanCMS22,DBLP:conf/eurocrypt/ZhouSCM23}.
We import this idea with a different emphasis: \randbin is used as an \emph{interface} that converts single-edit participation streams into a sparse, per-bin update stream with (i) a Hamming-style neighboring interface for downstream continual primitives and (ii) explicit backlog/delay guarantees calibrated by $(\varepsilon,\delta)$.
}

\textbf{Composition and Concurrency.}
Beyond classical $(\varepsilon,\delta)$-DP~\cite{DBLP:conf/icalp/Dwork06}, alternative formalisms (e.g., R\'enyi divergence and tradeoff-based views) can streamline composition reasoning~\cite{Miro17,dong2022gaussian,DBLP:conf/stoc/Vadhan023,DBLP:conf/innovations/ZhouZCS24}, recovering familiar advanced composition bounds~\cite{DBLP:conf/focs/DworkRV10,DBLP:conf/icml/KairouzOV15}.
Concurrent composition in interactive settings has also been studied, including adaptively chosen privacy parameters~\cite{DBLP:conf/ccs/HaneySTVVX023,Henzinger2026Continual}.
These works typically assume each mechanism’s neighboring relation is defined on a single underlying (static or dynamic) database; in our pipeline, we additionally need to reason about \emph{neighbor-preserving transformations} between components (edit-to-Hamming shielding), which motivates our modular composition viewpoint.
Our modular composition theorem can be viewed as combining the refinement-pair reduction of NPDP to standard DP in~\cite{DBLP:conf/innovations/ZhouZCS24} with the interactive/concurrent composition framework of~\cite{Henzinger2026Continual}; for completeness and to match our power-function accounting, we give a self-contained proof in Appendix~\ref{sec:list_composition}.

\medskip
\noindent \textbf{Conflict of Interest Disclosure.}
The authors declare no financial or other substantive conflicts of interest relevant to this work. 

\section{Setting and Preliminaries}
\label{sec:setting}

We fix the interactive streaming-learning model underlying continual release.
The adversary observes both the released model trajectory and release timing,
and may influence future user events. We also define participation adjacency
and the utility--latency metrics used throughout. Section~\ref{sec:overview}
then summarizes our buffering--aggregation pipeline and main guarantee.

\paragraph{Streaming Learning Loop and Transcript.}
At each step $t=1,2,\ldots$, an adaptive environment produces a
\emph{user event} $f_t$ (e.g., a private loss or update). Starting from a
public, data-independent model $w^{(0)}\in\mathbb{R}^d$, the server maintains
models $w^{(t)}$ and may release intermediate snapshots. To capture both
per-step and buffered releases, let $v_t\in\{0,1\}$ be the observable timing
bit: if $v_t=1$, the server incorporates available buffered events, forms an
update $g_t$, and publishes the new snapshot; if $v_t=0$, it holds
$w^{(t)}=w^{(t-1)}$. The adversary's view through horizon $T$ is
\[
\mathsf{tr}_{\leq T}
:=
\left(w^{(0)},(v_t,w^{(t)})_{t=1}^{T}\right).
\]
Our privacy notion protects this joint transcript, including timing, with the
same parameters for every horizon $T$.

\paragraph{Feedback/Adaptive Interaction.}
The environment, and hence the event stream, may depend on the past
transcript: $f_t$ may be an arbitrary randomized function of
$\mathsf{tr}_{<t}$. All guarantees quantify over such adaptive interaction.
A \emph{non-adaptive} analysis of a continual primitive instead assumes that
its input stream is fixed independently of past releases.

\paragraph{Participation Privacy via Single-Edit Adjacency.}
For event streams $F=(f_1,f_2,\ldots)$ and
$F'=(f'_1,f'_2,\ldots)$, write
$F\sim_{\mathrm{edit}}F'$ if one is obtained from the other by inserting or
deleting one event, shifting all subsequent indices. Equivalently, for some
$t_0$, either
\[
f'_t=f_t\quad(t<t_0),
\qquad
f'_t=f_{t+1}\quad(t\geq t_0),
\]
or the reverse relation holds. Thus the unit of privacy is one participation
event.

\paragraph{Differential Privacy Under Adaptive Interaction.}
A fixed-stream DP inequality alone does not capture feedback. Formally,
consider any adaptive paired environment $\mathcal{A}$ that, from the past
observable view, generates two candidate event prefixes satisfying the
prescribed edit-neighbor relation. Let
$\mathsf{View}^{\mathcal{A}}_{b,\leq T}$ denote its complete view when the
mechanism receives candidate $b\in\{0,1\}$. We require, for every
$\mathcal{A}$, horizon $T$, and measurable set $S$,
\[
\Pr\!\left[
\mathsf{View}^{\mathcal{A}}_{0,\leq T}\in S
\right]
\leq
e^\varepsilon
\Pr\!\left[
\mathsf{View}^{\mathcal{A}}_{1,\leq T}\in S
\right]
+\delta,
\]
and likewise with $0$ and $1$ exchanged. Fixed neighboring streams are the
special case in which $\mathcal{A}$ ignores past releases. Appendix~\ref{sec:prelim}
formalizes this paired simulation and its equivalent power-function
formulation.

\paragraph{Metrics.}
\emph{Utility:} test accuracy or loss of the released model $w^{(T)}$.
\emph{Systems:} backlog $Q_t$, the number of pending events after time $t$,
and inclusion delay
\[
D(i):=
\min\{t\geq i:\ f_i\text{ is incorporated by time }t\}-i,
\]
measured in event arrivals.

\paragraph{Extension to Multiple Edits.}
Adaptive group privacy applies unchanged to paired simulations: an adaptive
one-hop $(\varepsilon,\delta)$ guarantee yields, for distance at most $k$,
\[
\left(
k\varepsilon,\;
\delta\sum_{j=0}^{k-1}e^{j\varepsilon}
\right)\text{-DP}
\]
(Fact~\ref{fact:multi_hop_DP}). In particular, the two adjacent Hamming
changes that a single edit may induce at our downstream interface are charged
by this adaptive two-hop bound; the downstream one-hop parameters are
calibrated accordingly in Section~\ref{sec:overview}.

\section{Overview}
\label{sec:overview}

This section summarizes our buffering--aggregation pipeline (Fig.~\ref{fig:pipeline}) and the resulting end-to-end guarantee:
continual-release $(\varepsilon,\delta)$-DP for single-edit participation streams under feedback, where the adversary observes
the released trajectory (and timing). The required buffering level $U(\varepsilon,\delta)$ is privacy-implied and explicitly
governs delay.

\begin{figure}[H]
\centering
\resizebox{\columnwidth}{!}{%
\begin{tikzpicture}[
  font=\scriptsize,
  >=Latex,
  node distance=7mm and 9mm,
  block/.style={draw, rounded corners, align=center, inner sep=3pt, minimum height=8mm},
  small/.style={font=\scriptsize\itshape},
  arr/.style={->, line width=0.6pt}
]
\node[block] (stream) {User-event stream\\[-1pt] $f_1,f_2,\dots$};

\node[block, right=of stream] (buffer) {Randomized\\buffering\\[-1pt]
\small $U(\varepsilon,\delta)$};

\node[block, right=of buffer] (bins) {Batches\\[-1pt]
$B_t \subseteq \{f_i\}$\\[-1pt]
\small $|B_t|\in[U,2U]$};

\node[block, right=of bins] (agg) {Continual DP\\primitive\\[-1pt]
\small (certified)};

\node[block, right=of agg] (release) {Released models\\[-1pt]
$w^{(t)}$};

\draw[arr] (stream) -- (buffer);
\draw[arr] (buffer) -- (bins);
\draw[arr] (bins) -- node[above, small] {\shortstack{Update\\ $g_t$}} (agg);
\draw[arr] (agg) -- (release);

\draw[arr] ($(release.north)+(0,2mm)$) .. controls +(0,12mm) and +(0,12mm) ..
node[above, small, align=center] {feedback:\\future events may depend on $w^{(t)}$}
($(stream.north)+(0,2mm)$);

\node[small, below=3mm of release, align=center] (obs)
{adversary observes\\$w^{(t)}$};

\draw[-, line width=0.4pt] (release.south) -- (obs.north);

\node[small, below=3mm of bins, align=center] {bin contents / internal state hidden};
\end{tikzpicture}%
}
\vspace{-0.5em}
\caption{\textbf{Pipeline overview.} Randomized buffering (parameter $U(\varepsilon,\delta)$) converts a single-edit stream into bins that are consumed
(in the sense defined in Appendix~\ref{sec:modular_comp}) by a certified continual-release DP primitive, producing intermediate model releases. The adversary observes the released models; bin contents and internal states are hidden.}
\label{fig:pipeline}
\vspace{-0.8em}
\end{figure}
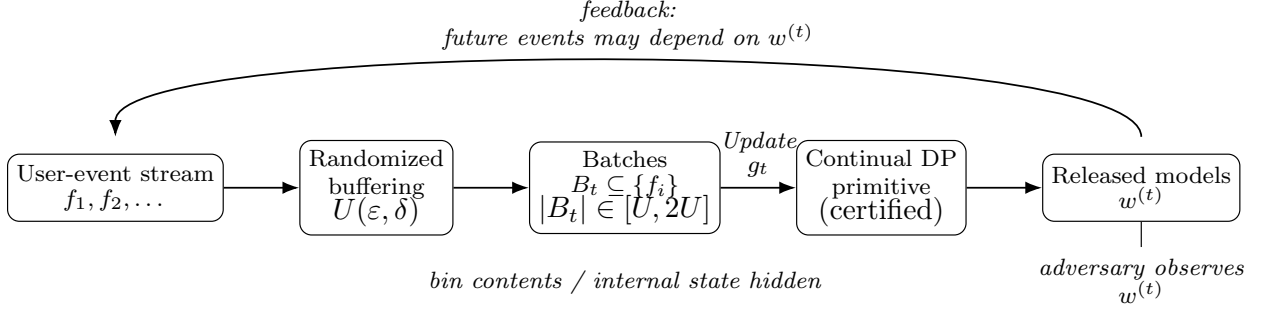


We next give the recipe and state the end-to-end guarantee.

\paragraph{Recipe (buffering $\rightarrow$ aggregation $\rightarrow$ releases).}
Given user-event stream $(f_t)$,
at each step $t\ge 1$, buffering outputs $B_t$ with $|B_t|\in[U,2U]$ or $\bot$.
Let $\tau(1)<\tau(2)<\cdots$ be the (random) steps with $B_{\tau(k)}\neq\bot$.

\begin{enumerate}\setlength{\itemsep}{0.1em}\setlength{\topsep}{0.1em}
\item \textbf{Buffer.} Run randomized buffering with $U(\varepsilon,\delta)$ to obtain $(B_t)_{t\ge 1}$.
\item \textbf{Update.} For each $k\ge 1$, form a clipped per-bin update
$
g_k := \frac{1}{|B_{\tau(k)}|}\sum_{f\in B_{\tau(k)}} \textsf{Clip}_G \Bigl(\!\nabla f\!\bigl(w^{(\tau(k)-1)}\bigr) \Bigr)
$, where
$\textsf{Clip}_G(x)
~:=~
x \cdot \min\!\left\{1,\ \frac{G}{\|x\|_2}\right\}$.

\item \textbf{Aggregate + release.} Feed $(g_k)_{k\ge 1}$ to a certified continual-release DP primitive and update only at bin times:
output $w^{(\tau(k))}$ after processing $g_k$, and hold $w^{(t)}=w^{(\tau(k))}$ for $\tau(k)\le t<\tau(k{+}1)$ (feedback allowed).
\end{enumerate}

\noindent Theorem~\ref{thm:main_pipeline_overview} states the resulting end-to-end privacy. 

\begin{theorem}[End-to-end continual-release DP for single-edit streams]
\label{thm:main_pipeline_overview}
Run the above recipe with privacy target $(\varepsilon,\delta)$, allocating
$(\varepsilon_b,\delta_b)$ to buffering and
$(\varepsilon_a,\delta_a)$ to the total downstream cost, where
$
\varepsilon_b+\varepsilon_a\leq\varepsilon,
\delta_b+\delta_a\leq\delta.
$
Set
\[
U=\Theta\!\Bigl(\frac{1}{\varepsilon_b}
\log\frac{1}{\delta_b}\Bigr),
\qquad
\widehat\varepsilon_a=\frac{\varepsilon_a}{2},
\qquad
\widehat\delta_a=
\frac{\delta_a}{1+e^{\varepsilon_a/2}}.
\]
Plug in any certified continual-release DP primitive that is adaptively
$(\widehat\varepsilon_a,\widehat\delta_a)$-DP for ordinary distance-one
Hamming-neighboring update streams. Then the released model trajectory
$(w^{(t)})_{t\geq 0}$ is $(\varepsilon,\delta)$-DP with respect to
\emph{single-edit} neighboring event streams, even when future events are
generated adaptively from past releases (feedback).
\end{theorem}

\noindent\emph{Proof idea.}
Randomized buffering turns a single insertion/deletion into at most two
adjacent Hamming changes in the emitted update stream. Certification gives
adaptive distance-one privacy, and adaptive group privacy
(Fact~\ref{fact:multi_hop_DP}) charges the two changes within the allocated
$(\varepsilon_a,\delta_a)$ budget. Modular composition with buffering then
gives the claimed guarantee (Appendix~\ref{sec:list_composition}).

\paragraph{Certification preview (adaptive safety of continual primitives).}
We use a reusable certification for feedback.  Writing $x_t$ for the
per-round input and $x_{1:t}$ for its prefix, the certified core must use
fresh independent randomness at each round and satisfy a stable-context
condition: common adaptive inputs may change the released statistic, but they
must not change 
the one-round privacy cost of a distance-one Hamming discrepancy.
Under these conditions, a standard non-adaptive
$(\varepsilon,\delta)$-DP analysis lifts to the feedback setting.
Section~\ref{sec:cert} gives the main-text checklist and theorem, while
Appendix~\ref{sec:ind_decomp} formalizes the stable-context condition as
quantitative decomposability.  The certification covers tree/prefix-sum
mechanisms and more general lower-triangular linear mechanisms.

\paragraph{SGD instantiation and what we measure.}
We instantiate the pipeline with streaming SGD: each emitted bin $B_{\tau(k)}$ triggers one clipped mini-batch gradient update,
privatized by a continual-release DP aggregator (tree/prefix-sum).
The target $(\varepsilon,\delta)$ fixes $U(\varepsilon,\delta)$, yielding random batch sizes in $[U,2U]$ and a backlog bound of order
$O\!\bigl(\tfrac{1}{\varepsilon_b}\log t\,(\log t+\log\tfrac{1}{\delta_b})\bigr)$ (Lemma~\ref{lem:randbin_backlog}).
The theoretical convergence of various SGD variants under DP noise has already been extensively analyzed; we refer the reader to standard works~\cite{DBLP:journals/jmlr/JainKT12,DBLP:conf/globalsip/SongCS13,DBLP:conf/nips/ThakurtaS13,DBLP:conf/icml/KairouzM00TX21,DBLP:conf/nips/DenisovMRST22}.

\paragraph{Deferred details.}
Sections~\ref{sec:randbin_short} and~\ref{sec:cert} formalize buffering and certification; the end-to-end proof of
Theorem~\ref{thm:main_pipeline_overview} (via modular composition) and additional instantiations/experiments appear in~\cite{DBLP:conf/icml/ChanSZZ26}.

\section{Auditable Ingredients: Randomized Buffering and Adaptive-Safe Continual-Release DP}
\label{sec:ingredients}

This section isolates the two interfaces used in
Theorem~\ref{thm:main_pipeline_overview}.
Randomized buffering (\randbin) supplies an edit-to-Hamming reduction and an
explicit backlog guarantee (Section~\ref{sec:randbin_short}), while
Theorem~\ref{thm:cert} certifies when a non-adaptive Hamming-DP continual
primitive remains private under feedback.
Section~\ref{sec:compose} composes these interfaces.

\ignore{
This section isolates the two ingredients that make the pipeline in Fig.~\ref{fig:pipeline} auditable.
First, we treat randomized buffering (\randbin) as an interface: it converts a \emph{single-edit} event stream
into a sparse per-bin update stream that admits a standard Hamming-neighbor DP analysis, while providing an explicit
backlog guarantee.
Second, we give a reusable certification theorem for adaptive interaction
(output-dependent updates): for continual-release mechanisms whose certified
cores use fresh independent randomness and have a stable one-round privacy
profile under common aggregation context, a non-adaptive DP proof remains
valid even when updates are chosen adaptively as a function of past releases.

\paragraph{Separation of Roles.}
Randomized buffering is not needed for the certification theorem itself:
Theorem~\ref{thm:cert} is a reusable adaptivity lift for continual mechanisms
that pass Checklist~1, combining fresh per-round randomness with stability of
the one-round privacy profile under common aggregation context.
In our pipeline, buffering serves a different
role: it shields the downstream primitive from single-edit adjacency, where one
insertion or deletion can shift many later batch boundaries. Thus, our modular
proof separates an upstream edit-to-Hamming interface from a downstream
adaptive-safety certificate. We do not claim that randomized buffering is the
only possible route; alternatives would require either another edit-shielding
interface or a direct native proof under edit adjacency.
}

\subsection{Randomized Buffering Wrapper}
\label{sec:randbin_short}

\paragraph{RandBin Interface.}
Given user events $f_1,f_2,\ldots$, at each step $t\ge 1$, \randbin outputs
either a bin $B_t$ of pending events or $\bot$ (no emission).
Let $\tau(1)<\tau(2)<\cdots$ denote the random emission times for which
$B_{\tau(k)}\neq\bot$. When $B_t=\bot$, nothing is sent downstream and the
released model remains unchanged. We use the following interface properties:
\begin{enumerate}\setlength{\itemsep}{0.05em}\setlength{\topsep}{0.05em}
  \item[\textbf{(P1)}] \textbf{Bin size.}
  Whenever $B_t\neq\bot$, $|B_t|\in[U,2U]$.

  \item[\textbf{(P2)}] \textbf{Backlog control.}
  The number $Q_t$ of pending events remains small for every $t$, yielding
  explicit inclusion-delay guarantees.

	  \item[\textbf{(P3)}] \textbf{Edit-to-Hamming shielding.}
  A single insertion/deletion induces coupled emitted-bin/update interface
  streams at Hamming distance at most two.

\end{enumerate}

\paragraph{Input-Identifiability Convention.}
Lemma~\ref{lem:randbin_wrapper} 
uses the standard convention that
events are identifiable, e.g., by unique IDs or metadata; anonymous duplicate
events require an additional disambiguation assumption that is discussed
in Appendix~\ref{sec:randbin}.

\begin{algorithm}[t]
\caption{\randbin (Interface-Level Skeleton)}
\label{alg:randbin_sketch}
\footnotesize
\begin{algorithmic}[1]
\STATE \textbf{Input:} $(\varepsilon_b,\delta_b)$; stream $(f_t)_{t\ge1}$
\STATE \textbf{Output:} $B_t\in\{\bot\}\cup\{\text{bins of size in }[U,2U]\}$
\STATE Set $U=\Theta(\varepsilon_b^{-1}\log(1/\delta_b))$ and a distribution $\mcal{D}$ supported on $[U,2U]$
\STATE Maintain a FIFO buffer $\mathsf{buf}$ and a private randomized emission scheduler
\FOR{$t=1,2,\dots$}
  \STATE Append $f_t$ to $\mathsf{buf}$
  \IF{there is no scheduled emission}
    \STATE Output $\bot$
  \ELSE
    \STATE Sample $C\in[U,2U]$ from $\mcal{D}$
     \STATE Pop $C$ items from $\mathsf{buf}$ as $B_t$
     \STATE Output $B_t$;
    refresh scheduler
  \ENDIF
\ENDFOR
\end{algorithmic}
\end{algorithm}

\ignore{
\begin{algorithm}[t]
\caption{\randbin (Interface-Level Skeleton)}\label{alg:randbin_sketch}
\footnotesize
\KwIn{$(\varepsilon_b,\delta_b)$; stream $(f_t)_{t\ge1}$}
\KwOut{$B_t\in\{\bot\}\cup\{\text{bins of size in }[U,2U]\}$}
Set $U=\Theta(\varepsilon_b^{-1}\log(1/\delta_b))$ and a distribution $\mcal{D}$ supported on $[U,2U]$\;
Maintain a FIFO buffer $\mathsf{buf}$ and a private (randomized) emission scheduler\;
\For{$t=1,2,\dots$}{
  append $f_t$ to $\mathsf{buf}$\;
  \If{no scheduled emission}{output $\bot$}
  \Else{sample $C\in[U,2U]$ from $\mcal{D}$; pop $C$ items from $\mathsf{buf}$ as $B_t$; output $B_t$; refresh scheduler}
}
\end{algorithm}
}

\paragraph{Backlog and Delay.}
The buffering level $U(\varepsilon_b,\delta_b)$ determines how aggressively
\randbin batches the stream. We use the backlog $Q_t$ and inclusion delay
$D(i)$ defined in Section~\ref{sec:setting}; because events are emitted from a
FIFO buffer, bounds on $Q_t$ yield corresponding delay bounds.


\ignore{
\paragraph{Systems Metrics: Backlog and Delay.}
The buffering level $U(\varepsilon_b,\delta_b)$ controls how aggressively \randbin batches the stream and thus the
latency regime seen by learning. We track the \textbf{backlog} $Q_t$, the number of pending (unbinned) events after time $t$.
For an event $f_i$, its \textbf{inclusion delay} $D(i)$ is the number of subsequent arrivals until it is placed in an emitted bin.
In particular, if $Q_t \le B(t)$ for all $t\in[i,T]$, then $D(i)\le B(T)$ (and sharper bounds follow from the trajectory of $(Q_t)$).


\paragraph{Interface Guarantees for Modular Composition.}
We will use two interface properties of \randbin; proofs and constants appear in Appendix~\ref{sec:randbin}.
}

\begin{lemma}[Edit-to-Hamming Wrapper (NPDP Contract)]
\label{lem:randbin_wrapper}
For any single-edit neighboring input streams, \randbin admits a
neighbor-preserving paired simulation (with refinement) in which the coupled
emitted-bin/update interface streams have Hamming distance at most two.
If two concrete-bin positions differ, they are adjacent.
\end{lemma}

\noindent\emph{Intuition:} randomized boundaries prevent one
insertion/deletion from cascading into many downstream bin shifts; the paired
simulation formalizes this stability under interaction.

\begin{lemma}[Backlog Bound (Privacy-Implied Delay)]
\label{lem:randbin_backlog}
With $U=\Theta(\varepsilon_b^{-1}\log(1/\delta_b))$, for every time $t$ the backlog satisfies
$
Q_t \;=\; O\!\Bigl(\frac{1}{\varepsilon_b}\log t\cdot\Bigl(\log t+\log\frac{1}{\delta_b}\Bigr)\Bigr).
$

This bound is obtained by implementing \randbin's private emission scheduler via the streaming DP prefix-sum mechanism
(Theorem~\ref{th:prefix_sum}) and translating its additive-error guarantee into a
worst-case queue bound.
As a consequence, inclusion delay is controlled at the same order (up to constants) for events arriving by time $t$.
\end{lemma}
\noindent\emph{Intuition:} the noisy scheduling forces sufficiently frequent emissions to prevent sustained queue growth, while keeping bin sizes within $[U,2U]$.

\paragraph{On the Horizon Dependence.}
The $\log t$ dependence in Lemma~\ref{lem:randbin_backlog} is not merely
algebraic. Our emission scheduler is a continual private prefix-sum mechanism
whose additive error translates directly into backlog. On a unit-rate stream,
if $M_t$ denotes cumulative emissions by time $t$, then $Q_t=t-M_t$. Hence,
within this scheduler family, improving backlog requires improving continual
private counting error. Known lower bounds suggest that logarithmic horizon
dependence is intrinsic to this proof route~\cite{DBLP:conf/stoc/DworkNPR10,
DBLP:conf/colt/Cohen0NSS24}, though we do not claim a minimax lower bound over
all NPDP wrappers.

\paragraph{Scope of the Approximate-DP Regime.}
The probability-one finite-delay guarantee in
Lemma~\ref{lem:randbin_backlog} relies on the approximate-DP implementation
of \randbin: its proof uses bounded-support buffering noise and an
always-valid prefix-sum error bound. RandBin-style shift masking under pure
DP requires unbounded-support noise and thus cannot provide the same
guarantee. 

\paragraph{Positioning.}
Randomized buffering is often used in oblivious data structures to hide
access patterns~\cite{DBLP:journals/jacm/ChanCMS22,DBLP:conf/eurocrypt/ZhouSCM23}.
Our emphasis differs: \randbin is an \emph{infrastructure interface} that
(i) converts single-edit participation streams into sparse, Hamming-style
per-bin update streams and (ii) provides explicit backlog guarantees.
Once Lemmas~\ref{lem:randbin_wrapper}--\ref{lem:randbin_backlog} hold, any
continual primitive certified in Section~\ref{sec:cert} can be reused
downstream. Full pseudocode, the NPDP security game (paired simulation and
refinement), and sharper tail bounds appear in
Appendix~\ref{sec:randbin}.

\subsection{Certifying Feedback-Safe Continual-Release Primitives}
\label{sec:cert}

\begin{center}
\fbox{\begin{minipage}{0.97\columnwidth}
\textbf{Checklist 1 (Certification Conditions).}
A continual-release primitive $\mathcal{M}$ passes certification if its core
release process satisfies the following conditions:
\begin{enumerate}\setlength{\itemsep}{0.15em}\setlength{\topsep}{0.15em}
  \item \textbf{Prefix dependence:} the core release at time $t$ depends only
  on the input prefix $x_{1:t}$, not on future inputs.

  \item \textbf{Fresh randomness:} the core release at time $t$ uses
  randomness $R_t$ that is independent of the randomness used in all other
  rounds.

  \item \textbf{Deterministic release map:} given $(R_t,x_{1:t})$, the core
  release $y_t=\mathcal{M}_t(R_t,x_{1:t})$ is deterministic.

  \item \textbf{Stable one-round privacy cost:} if two prefixes differ in one
  update, then common aggregation context in the two prefixes does not change
  the privacy cost of that discrepancy.  The appendix formalizes this
  condition as quantitative decomposability in Definition~\ref{defn:quant_decomp}.

  \item \textbf{Post-processing only afterward:} after the certified core
  releases are produced, the mechanism may output any function of the released
  transcript; this is post-processing and does not weaken the privacy
  guarantee.
\end{enumerate}
\emph{Consequence:} any finite-horizon non-adaptive $(\varepsilon,\delta)$-DP
proof for Hamming-neighboring update streams lifts to adaptive interaction.
\end{minipage}}
\end{center}

\paragraph{Why Certification?}
Our pipeline treats the continual-release DP primitive as a plug-in component.
To make this modularity auditable, we need a criterion under which a standard
non-adaptive Hamming-neighbor DP analysis remains valid when per-round updates
are chosen adaptively from past releases. Fresh independent randomness is
necessary but not sufficient; Appendix~\ref{subsec:ind_decomp_counterexample}
gives a two-round counterexample. Our criterion therefore adds quantitative
decomposability, preventing adaptive common context from changing the one-round
leakage profile.

\noindent \paragraph{Certified Decomposability.}
Write $x_t$ for the input update at round $t$ and
$x_{1:t}:=(x_1,\dots,x_t)$.  The certification condition has two parts.
First, the mechanism must use fresh independent randomness at each round:
the release at time~$t$ is generated as
\[
y_t=\mathcal{M}_t(R_t, x_{1:t}),
\]
where $R_t$ is independent of the randomness used at other rounds and
$\mathcal{M}_t$ is deterministic given $(R_t, x_{1:t})$.

Second, the privacy cost of a one-round release must be stable under common
context.  Informally, if two input histories differ in one update, then the
effect of this discrepancy on the round-$t$ release is determined only by the
part of the round-$t$ aggregate that contains the changed update.  Common
updates before or after the discrepancy may change the released statistic, but
they are shared by both neighboring histories and therefore should not change
the privacy cost of the discrepancy itself.

The formal appendix condition, called quantitative decomposability, rules out
this behavior by requiring the one-round privacy profile to depend only on the
distance between the relevant aggregates after common aggregation context is
removed.

\emph{Importantly,} after producing the output $y_{1:t}$ at time~$t$, the mechanism may apply \emph{arbitrary} post-processing to $(y_1,\dots,y_t)$ and output any derived value at time $t$; by closure under post-processing, this
does not weaken the privacy guarantee proved for the underlying releases.

\paragraph{Certification Theorem.}
The next theorem formalizes the consequence of Checklist~1 and serves as an
audit tool for reusing continual DP primitives inside interactive learning
loops.  The formal appendix statement is given in terms of power functions;
the version below states the resulting $(\varepsilon,\delta)$-DP guarantee.

\begin{theorem}[Certification: Non-Adaptive DP Implies Feedback-Safe DP]
\label{thm:cert}
Consider Hamming-style neighboring update streams, where two equal-length
streams differ in at most one position.  Let $\mathcal{M}$ be a continual
primitive satisfying Checklist~1.  If $\mathcal{M}$ is non-adaptively
$(\varepsilon,\delta)$-DP for Hamming-neighboring update streams for every
finite horizon, then $\mathcal{M}$ is also $(\varepsilon,\delta)$-DP against
adaptive interaction for every finite horizon, where each $x_t$ may be chosen
as an arbitrary function of past releases.
\end{theorem}

\paragraph{Proof Sketch.} (Full proof in Appendix~\ref{sec:ind_decomp})
The proof works in the paired-simulation game and proceeds by induction on
the horizon.  If the first input pair already contains the unique Hamming
discrepancy, then all later paired inputs must be equal componentwise.
Checklist~1 ensures that these later common inputs only add common context and
do not change the one-round privacy cost of the discrepancy; the per-round
privacy losses then tensorize exactly.  If the first input pair is equal, we
condition on the common first release, fix the first input as common context,
and apply the induction hypothesis to the shifted mechanism on the remaining
rounds.  Fresh independent randomness is used to separate the per-round
contributions, while the stable-context condition prevents adaptive common
inputs from amplifying the leakage of the unique discrepancy.

\paragraph{Canonical Examples.}
For the tree mechanism, take the certified core to be the newly revealed noisy
tree-node values. Each uses fresh noise, and only nodes whose intervals contain
the changed update have different aggregates; published prefix sums are
post-processing of these values. For a lower-triangular mechanism, take as the
core
\[
z_t=\langle b_t,x_{1:t}\rangle+\eta_t,
\qquad
\supp(b_t)\subseteq\{1,\ldots,t\},
\]
with fresh independent $\eta_t$. Common updates contribute equally to the two
neighboring executions, and any output
$y_t=\sum_{s\le t}A_{t,s}z_s$ is post-processing. Thus both mechanisms pass
Checklist~1; details appear in Appendix~\ref{sec:prefixsum_factor}.

\paragraph{Scope Boundary.}
Theorem~\ref{thm:cert} is a sufficient certification condition, not a
characterization of adaptive privacy.  Fresh randomness alone is not enough;
Appendix~\ref{subsec:ind_decomp_counterexample} gives a two-round separation.
The missing ingredient is the stable-context condition in Checklist~1, which
prevents adaptive common inputs from changing the leakage profile of an
earlier Hamming discrepancy.  Conversely, failure of Checklist~1 does not
imply that adaptive privacy is false; it means that a separate
interaction-aware proof is needed, as for \randbin in
Appendix~\ref{sec:randbin}.

\subsection{Putting the Ingredients Together}
\label{sec:compose}

We view \randbin as producing a visible buffering transcript and a hidden
per-bin update stream consumed by the downstream continual primitive.
By Lemma~\ref{lem:randbin_wrapper}, under a single input edit the two coupled
interface streams have Hamming distance at most two.

Let the downstream primitive pass Checklist~1 and be non-adaptively
$(\widehat\varepsilon_a,\widehat\delta_a)$-DP for ordinary distance-one
Hamming neighbors. Theorem~\ref{thm:cert} gives the same guarantee under
adaptive interaction, and adaptive group privacy
(Fact~\ref{fact:multi_hop_DP}) gives the distance-two guarantee
$
\left(
2\widehat\varepsilon_a,\,
(1+e^{\widehat\varepsilon_a})\widehat\delta_a
\right).
$
Accordingly, we set
$
\widehat\varepsilon_a=\frac{\varepsilon_a}{2},
\widehat\delta_a=
\frac{\delta_a}{1+e^{\varepsilon_a/2}},
$
so the total downstream cost is at most
$(\varepsilon_a,\delta_a)$.

Finally, composing \randbin's
$(\varepsilon_b,\delta_b)$-NPDP guarantee with this downstream guarantee
using the modular composition theorem in
Appendix~\ref{sec:list_composition} yields
$
(\varepsilon_b+\varepsilon_a,\,
  \delta_b+\delta_a)\text{-DP}
$
for single-edit neighboring event streams, even under adaptive interaction.
This proves Theorem~\ref{thm:main_pipeline_overview}.

\paragraph{Takeaway.}
Once the wrapper and certification conditions hold, the downstream primitive
can be reused after calibrating its distance-one budget for the two-hop
interface.

\section{Instantiation: Continual DP-SGD Pipeline}
\label{sec:instantiation}

This section instantiates the buffering--aggregation recipe
(Section~\ref{sec:ingredients}) with streaming SGD
(Algorithm~\ref{alg:dp_sgd_edit}). Each arriving user event $f_t$ is buffered
by \randbin. When a bin is emitted, we form a clipped mini-batch gradient and
pass a stepsize-scaled update to a continual-release DP aggregator. 
The aggregator \PStreamSum is either (i)~the standard tree-based
prefix-sum mechanism~\cite{DBLP:conf/icalp/ChanSS10}, or (ii)~the
lower-triangular matrix factorization prefix-sum
mechanism~\cite{DBLP:conf/icml/FichtenbergerHU23}. Their formal DP and
error/consistency guarantees are summarized in
Appendix~\ref{sec:prefix-sum}. As in Fig.~\ref{fig:pipeline}, the adversary
observes the released models and release timing (i.e., whether $B_t=\bot$ at
each step), while bin contents and internal states are hidden.

\begin{algorithm}[t]
\caption{Continual DP-SGD Under Single-Edit User Streams}
\label{alg:dp_sgd_edit}
\footnotesize
\begin{algorithmic}[1]
\STATE \textbf{Input:} Privacy target $(\varepsilon,\delta)$; public horizon
$T_{\max}$, i.e., an upper bound on both the number of time indices and emitted
bins; allocation $(\varepsilon_b,\delta_b)$ to buffering and
$(\varepsilon_a,\delta_a)$ to the total downstream cost, with
$\varepsilon_b+\varepsilon_a\leq\varepsilon$ and
$\delta_b+\delta_a\leq\delta$; initial model $w^{(0)}$; clip norm $G$;
stepsizes $(\eta_k)_{k\geq1}$
\STATE \textbf{Output:} Released trajectory
$(w^{(t)})_{t=0}^{T_{\max}}$ and release timing
$\bigl(\mathbf{1}\{B_t\neq\bot\}\bigr)_{t=1}^{T_{\max}}$
\STATE Set $U=\Theta(\varepsilon_b^{-1}\log(1/\delta_b))$ and initialize
\randbin with level $U$
\STATE Set
$\widehat\varepsilon_a\gets\varepsilon_a/2$,
$\widehat\delta_a\gets
\delta_a/(1+e^{\varepsilon_a/2})$, and
$\Delta_u\gets 2G\max_{1\leq j\leq T_{\max}}\eta_j$
\STATE Initialize \PStreamSum with one-hop budget
$(\widehat\varepsilon_a,\widehat\delta_a)$ and
$\ell_2$-sensitivity $\Delta_u$
\STATE $k\gets 0$
\FOR{$t=1,2,\dots,T_{\max}$}
  \STATE Receive user event $f_t$
  \STATE $B_t \gets \randbin(f_t)$
  \IF{$B_t=\bot$}
    \STATE $w^{(t)} \gets w^{(t-1)}$
  \ELSE
    \STATE $k\gets k+1$
    \STATE $g_k \gets \frac{1}{|B_t|}\sum_{f\in B_t}
    \textsf{Clip}_G\!\Bigl(\nabla f\bigl(w^{(t-1)}\bigr)\Bigr)$
    \STATE $u_k \gets \eta_k\,g_k$
    \COMMENT{stepsize-scaled update}
    \STATE $S_k \gets \PStreamSum(u_k)$
    \COMMENT{private prefix sum}
    \STATE $w^{(t)} \gets w^{(0)}-S_k$
    \COMMENT{post-processing of $S_k$}
  \ENDIF
\ENDFOR
\end{algorithmic}
\end{algorithm}

\paragraph{Fixed vs.\ Tuned Parameters.}
The privacy allocation determines the buffering level
$U(\varepsilon_b,\delta_b)$ and reserves
$(\varepsilon_a,\delta_a)$ for the total downstream cost. Since a single edit
can alter at most two adjacent update positions, \PStreamSum is run with the
one-hop parameters
\[
\widehat\varepsilon_a=\frac{\varepsilon_a}{2},
\qquad
\widehat\delta_a=
\frac{\delta_a}{1+e^{\varepsilon_a/2}}.
\]
Given a learning task, the clipping norm $G$ and stepsizes $(\eta_k)$ are tuned
for utility; once chosen, they determine the conservative one-hop update
sensitivity
\[
\Delta_u=2G\max_{1\leq k\leq T_{\max}}\eta_k.
\]
The release schedule is induced by \randbin, with
$|B_t|\in[U,2U]$ whenever $B_t\neq\bot$. The plug-in aggregator must be
non-adaptively
$(\widehat\varepsilon_a,\widehat\delta_a)$-DP for ordinary distance-one
Hamming neighbors at sensitivity $\Delta_u$ and pass Checklist~1.
Theorem~\ref{thm:cert} lifts this guarantee to adaptive interaction, and
adaptive group privacy (Fact~\ref{fact:multi_hop_DP}) bounds the resulting
two-hop cost by $(\varepsilon_a,\delta_a)$.

\paragraph{Privacy-Budget Split.}
A symmetric $50/50$ split between buffering and the total downstream cost is a
simple default. More generally, any fixed split
$\varepsilon_b=c\varepsilon$ and
$\varepsilon_a=(1-c)\varepsilon$, with $c\in(0,1)$, changes the asymptotic
guarantees only by constant factors, while assigning constant fractions of
$\delta$ changes $\log(1/\delta_b)$ only by an additive $O(1)$ term. The best
finite-sample split can be selected by sweeping $c$ on public or separately
privatized validation data and measuring the latency--accuracy tradeoff.

\paragraph{Mini-Batch View.}
When a bin is emitted, the gradient estimate is the clipped mini-batch average
\[
\frac{1}{|B_t|}
\sum_{f\in B_t}\textsf{Clip}_G\!\bigl(\nabla f(w)\bigr).
\]
Thus $\|g_k\|_2\leq G$. Moreover, \randbin enforces
$|B_t|\in[U,2U]$, so the batch size varies by at most a factor of $2$.

\paragraph{Implementation Notes and Measured Latency.}
The server maintains (a) the current released model $w^{(t)}$, (b) the FIFO
buffer state of \randbin, and (c) the internal state of \PStreamSum
(e.g., its noisy tree or factorization state). The model is updated only at
emission times and is held fixed otherwise, yielding a piecewise-constant
released trajectory. To quantify privacy-implied latency, we log the backlog
$Q_t$ of pending events after time $t$. For an event $f_i$, its inclusion delay
is
\[
D(i):=\min\{t\geq i:\ f_i\in B_t\}-i,
\]
which can be computed from the same FIFO queue evolution underlying $(Q_t)$.

\paragraph{Other Optimizers.}
The same template applies to streaming optimizers whose per-bin updates can be
encoded as a clipped, sensitivity-bounded stream for a certified
continual-release aggregator: replace the definitions of $g_k$ and $u_k$ while
keeping \randbin as the edit-to-Hamming wrapper and \PStreamSum as the
aggregation primitive. Variants with additional state, such as momentum or
adaptive preconditioning, require a separate audit of
Theorem~\ref{thm:cert}'s conditions; falling outside the theorem does not by
itself imply failure of adaptive privacy. 

%
%
%

\section{Conclusion}
\label{sec:conclusion}

We studied continual learning under adaptive interaction when privacy is defined over \emph{edit-style} (single-insertion/deletion)
user streams, where naive continual-release analyses do not directly apply.
Our main recipe composes a randomized buffering wrapper \randbin, which reduces single-edit neighbors to a Hamming-style
per-bin update interface with explicit backlog guarantees, with a certified continual-release DP primitive (tree-based
prefix sums in our experiments), and then invokes modular composition to obtain end-to-end $(\varepsilon,\delta)$-DP.


\newpage

\bibliographystyle{alpha}
\bibliography{f-DP,admm,oram}

\appendix
\section*{Appendix}

\section{Detailed Preliminaries}
\label{sec:prelim}

We formalize the setting as outlined in the introduction.
We first describe the concept of
an \emph{abstract} interactive mechanism,
from which we will derive other forms of interactive mechanisms later.

\paragraph{Finite-Precision Convention.}
The simulation framework below is stated for finite message and seed spaces.
For every public finite horizon $T$, we interpret all communicated inputs,
states, and released numerical quantities using fixed finite-precision
encodings. Whenever we write a continuous or unbounded-noise mechanism, the
formal mechanism is obtained by applying a fixed, data-independent
quantization and clipping map to the ideal noisy release. The resulting
finite-output kernels admit finite-seed representations. By data processing,
all stated DP, RDP, and zCDP upper bounds for the ideal releases continue to
hold. We suppress this representation layer throughout.

\paragraph{Abstract Interactive Mechanism.}
We use $\mcal{U}$
to denote the collection of all valid messages,
where $\{\bot, \halt\} \subseteq \mcal{U}$.
We use $\mcal{R}$ to denote the collection of random seeds.
In this work, we focus on the case that both $\mcal{U}$ and $\mcal{R}$ are finite.

An interactive mechanism~$\algM$ is specified
by a distribution $R_\algM$ on $\mcal{R}$ and a transition function
 $\algM: \mcal{R} \times \mcal{U}^* \to \mcal{U}$,
where
$\mcal{U}^*:=\cup_{i \in {\mathbb{N}}} \, \mcal{U}^i$,
with the convention that $\mcal{U}^0$ is
the identity under Cartesian product, i.e., $\mcal{R} \times \mcal{U}^0 = \mcal{R}$.
Note that we sometimes overload the notation such that $\algM$
denotes both the mechanism and the corresponding transition function.
Here is the functionality of an interactive mechanism.

\begin{tcolorbox}
\begin{enumerate}[leftmargin=1pt]
\item The mechanism $\algM$ samples a (secret) random seed $r \in \mcal{R}$
according to the distribution~$R_\algM$.

\item At time step $t = 0$, the mechanism outputs $\algM(r) \in \mcal{U}$.
If the mechanism is supposed to wait for the first input message,
we use the convention that $\algM(r) = \bot$.

\item At time step $t \geq 1$, suppose a message
$x_t \in \mcal{U}$ (which may be chosen depending on the history before time~$t$) is sent to the
mechanism.  We use $x[1..t] 
:= (x_1, x_2, \ldots, x_t) \in \mcal{U}^t$
to denote the messages received by the mechanism from the previous time steps
up to~$t$.

Then, at this time step, the mechanism outputs $\algM(r; x[1..t]) \in \mcal{U}$.

\item If, for some $t$, $\algM(r; x[1..t]) = \halt$,
then the mechanism halts at time step~$t$.

\end{enumerate}
\end{tcolorbox}

\paragraph{Instantiation: Mapping the Abstract Interaction to Our Transcript.}
The abstract interaction formalism above subsumes the continual-learning setting of Section~\ref{sec:setting}.
Concretely, we instantiate the mechanism's input messages as the user-side stream $x_t:=f_t$ (or, when modeling a
subcomponent, $x_t$ may represent an emitted bin $B_t$ or a per-bin update $g_t$), and we instantiate the mechanism’s
outputs as the released observations $y_t:=(v_t,w^{(t)})$, so that the induced transcript is
$\mathsf{tr}_{\le T}=((v_t,w^{(t)}))_{t=0}^T$.
The adversary's view function $\nu_{\algA}$ reveals exactly the observable portion of each release---the timing bit and
model snapshot---matching the trajectory-and-timing threat model in Section~\ref{sec:setting}.

\noindent \textbf{Bounded Termination.}
In addition to finite $\mcal{R}$ and $\mcal{U}$,
we consider mechanisms that always terminate in a bounded number of steps.
Specifically, for each mechanism $\mcal{M}$,
there exists $T_{\algM} \in \Z$ such that $\algM$ always terminates in
at most $T_{\algM}$ steps.  This is without loss of generality,
because any privacy notion defined for mechanisms with bounded termination
can be extended naturally to unbounded mechanisms.  We simply require that, for any $T \in \Z_+$, the truncated mechanism obtained by running the (unbounded) mechanism for $T$ steps will satisfy the privacy notion defined for bounded mechanisms.
This matches the ``for every horizon $T$'' convention in Section~\ref{sec:setting}.

\noindent \emph{Mechanism View.} The view observed by a mechanism $\algM$ up to time step~$t$
consists of its generated random seed~$r$ and
its received input messages~$x[1..t]$ up to time~$t$.

\begin{remark}
Since $\algM$ can recover its own output sequence from its random seed and input sequence, it suffices to include the latter two in its view.
\end{remark}

\noindent \textbf{Similarity Between Mechanisms.}
To formalize privacy notions later,
we need a way to quantify the similarity between two mechanisms.
Since we consider randomized mechanisms,
we will use power functions to compare how different two distributions
are.  It is known that power functions are general enough
to capture all divergences satisfying the data processing inequality.
Recall that in this work, we focus on finite sample spaces.

\begin{definition}[Data Processing Inequality]
\label{defn:dpi}
A divergence measure is a function 
$\D$ that takes two distributions $P$ and $Q$ on the same space such that
$\D(P \| Q) \geq 0$, where equality holds \emph{iff} the distributions $P = Q$ are identical.

A divergence $\D$ satisfies the \emph{data processing inequality} if, for any stochastic transformation (or channel) $T$ that maps the original space to another space, the following inequality holds:

$$\D(T(P) \| T(Q)) \leq \D(P \| Q),$$

where $T(P)$ and $T(Q)$
are the resulting distributions after applying the transformation $T$
to $P$ and $Q$, respectively.

\end{definition}

\begin{definition}[Power Function as Fractional Knapsack Problem~\cite{kadane1968discrete}]
\label{defn:power}
Suppose $P$ and $Q$ are distributions on the same finite sample space~$\Omega$,
i.e., $P$ and $Q$ are vectors in $\R^\Omega_{\geq 0}$ whose coordinates sum to 1.
The power function $\Pow(P \| Q): [0,1] \rightarrow [0,1]$ can be
defined in terms of the \emph{fractional knapsack problem}.

Given a collection $\Omega$ of items,
suppose $\omega \in \Omega$ has weight $P(\omega)$ and value $Q(\omega)$.
Then, given $\alpha \in [0,1]$,
$\Pow(P \| Q)(\alpha)$ is the maximum value attained with weight capacity
constraint~$\alpha$, where items may be taken fractionally.
\end{definition}

\noindent \emph{Intuition.}
When two distributions $P$ and $Q$
are the same, it is clear that given any capacity constraint~$\alpha$,
the maximum reward is also $\alpha$; this means the power function
is exactly the identity function.  However,
if the two distributions are very different,
this means that there are items whose reward-to-weight ratios are large;
hence, in this case, the power function can initially grow faster than
the identity function.

\noindent \textbf{Partial Order on Power Functions.}
Pointwise comparison naturally induces a partial order on
power functions.  We denote $f_1 \leq f_2$ if
for all $\alpha \in [0,1]$, $f_1(\alpha) \leq f_2(\alpha)$,
where a larger $f_2$ indicates that the two distributions
are more different.

\begin{fact}[Properties of Power Functions]
For any two distributions $P$ and $Q$ on the same sample space,
$\Pow(P \| Q) \geq \mathsf{Id}$, 
where $\mathsf{Id}: [0,1] \rightarrow [0,1]$ is the identity function,
and equality holds \emph{iff} the distributions $P = Q$ are identical.

Moreover, power functions satisfy the \emph{data processing inequality},
i.e., for any stochastic transformation (or channel) $T$ that maps the original space to another space, the following inequality holds:

$$\Pow(T(P) \| T(Q)) \leq \Pow(P \| Q).$$

\end{fact}

\begin{figure}[!ht]
\centering 
\includegraphics[width=0.3\textwidth]{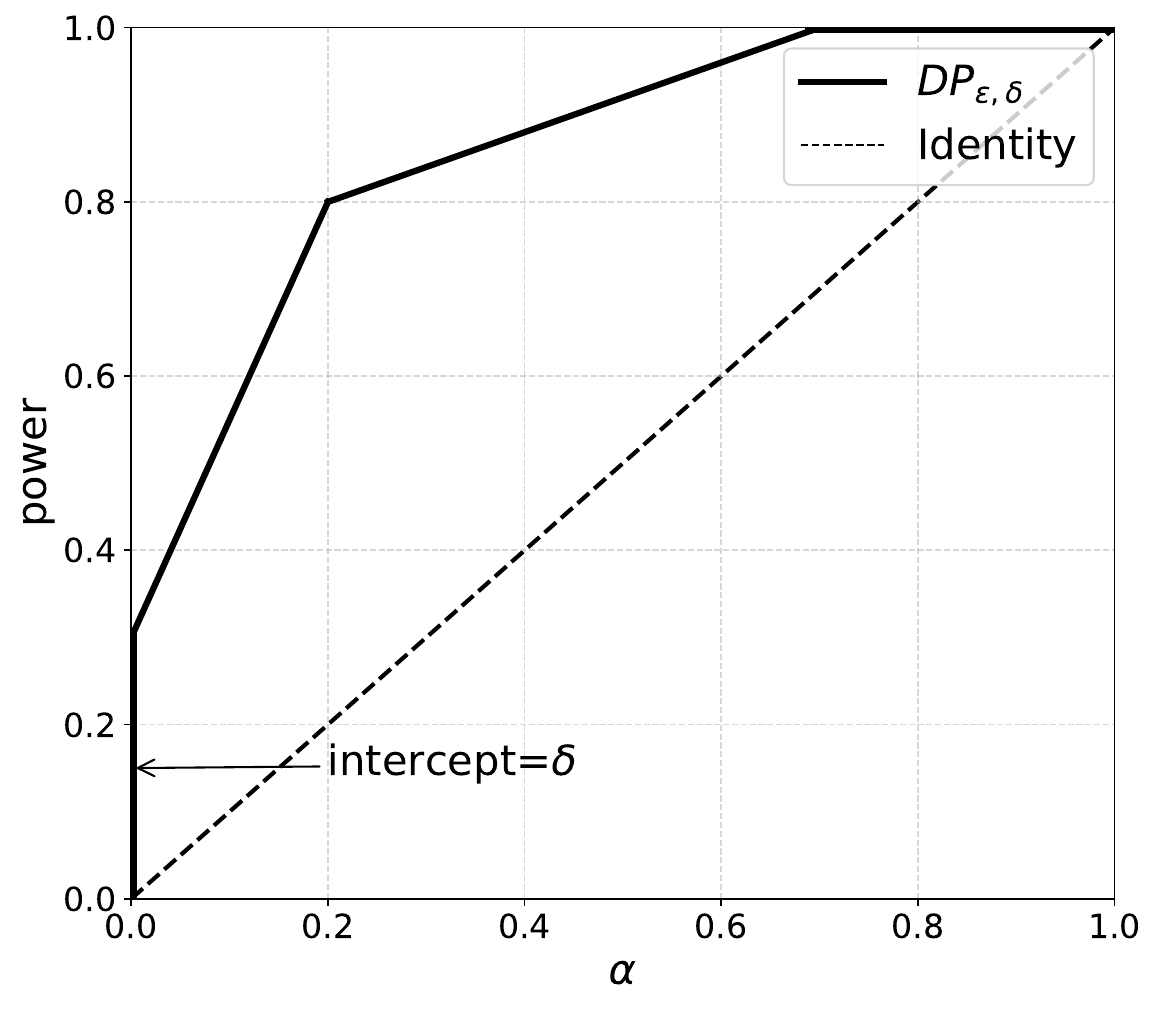}
\caption{Power function $\DP_{\varepsilon, \delta}$ for $\varepsilon = 1$ and $\delta = 0.3$.}
\label{fig:power_function}
\end{figure}

\noindent \textbf{Differential Privacy in Terms of Power Function.}
Given $\varepsilon \geq 0$ and $0 \leq \delta \leq 1$,
the power function $\DP_{\varepsilon, \delta}: [0, 1] \rightarrow [0,1]$
can be described as follows (see Figure~\ref{fig:power_function}):
\begin{enumerate}
\item In the $xy$-plane, starting at $(0, \delta)$, 
the function increases linearly with slope $e^\varepsilon$
until it touches the line $y = 1 - x$ at $(\frac{1-\delta}{1+ e^\varepsilon}, 
\frac{e^\varepsilon + \delta}{1+ e^\varepsilon})$.

\item After touching the line $y=1-x$,
the slope changes to $e^{-\varepsilon}$, until the line segment
reaches $(1 - \delta, 1)$, where the curve remains horizontally
till the end $(1,1)$ is reached.
\end{enumerate}
In other words,

$
\DP_{\varepsilon, \delta}(\alpha) = 
\begin{cases}
\delta + e^\varepsilon \cdot \alpha, & \text{for } 0 \leq \alpha \leq \frac{1-\delta}{1+ e^\varepsilon};\\
1 - e^{-\varepsilon} (1 - \delta) + e^{-\varepsilon} \alpha, & \text{for } \frac{1-\delta}{1+ e^\varepsilon} < \alpha \leq 1 - \delta; \\
1, & \text{for } 1 - \delta < \alpha \leq 1.
\end{cases}
$

Below are some facts on power functions that we will use,
which are  usually equivalently stated
in the literature
in terms of the tradeoff function
$\mathsf{T}(P \| Q)(\alpha) := 1 - \Pow(P\| Q)(\alpha)$.

\begin{fact}[DP in Terms of Power Function~\cite{dong2022gaussian}]
Given two distributions $P$ and $Q$ on the same sample space~$\Omega$,
$\Pow(P \| Q) \leq \DP_{\varepsilon, \delta}$ is equivalent to the statement that
for all $S \subseteq \Omega$,
$P(S) \leq e^\varepsilon \cdot Q(S) + \delta$ and 
$Q(S) \leq e^\varepsilon \cdot P(S) + \delta$.
\end{fact}

\begin{fact}(\emph{Triangle Inequality for Power Functions})[Theorem 2.14 in~\cite{dong2022gaussian}]
\label{fact:triangle_power}
Suppose $X$, $Y$, $Z$ are
distributions on the same sample space 
such that $\Pow(X \| Y ) \leq g_1$ and $\Pow(Y \| Z) \leq g_2$. 
Then,
$\Pow(X \| Z) \leq g_2 \circ g_1$, where the composition is defined as $(g_2 \circ g_1)(x) := g_2(g_1(x))$.
In other words, $\Pow(X \| Z) \leq \Pow(Y \| Z) \circ \Pow(X \| Y )$.
\end{fact}

\noindent \textbf{Example.}
Consider $g = \DP_{\varepsilon, \delta}$.
Then, $g \circ g \leq \DP_{2 \varepsilon, \delta(1+ e^\varepsilon)}$.
As we shall see, this is relevant to the scenario
of \emph{multi-hop} neighbors -- \emph{aka} \emph{group} DP in the literature.

\begin{definition}[Supremum of Power Functions]
\label{defn:power_sup}
Given a collection $\mcal{S}$
of power functions, its supremum $\sup(\mcal{S}) = \sup_{f \in \mcal{S}} f$
is the least power function~$\widehat{f}$ such that for all $\f \in \mcal{S}$,
$f \leq \widehat{f}$.

Formally, for each $\alpha \in [0,1]$, $\widehat{f}(\alpha) :=$
$$\sup \{\lambda \cdot f_1(\alpha_1) + (1- \lambda) \cdot f_2(\alpha_2) \, | \,
f_1, f_2 \in \mcal{S}; \alpha_1, \alpha_2, \lambda \in [0,1]:
\alpha = \lambda \alpha_1 + (1-\lambda) \alpha_2\}.$$
\end{definition}

\begin{definition}[Tensor Product of Power Functions]
\label{defn:tensor}
Suppose $f$ and $g$ are power functions
that are attained by the corresponding distributions in
$\Pow(P_1 \| P_2) = f$ and $g = \Pow(Q_1 \| Q_2)$.
Then,
their corresponding tensor
product $f \otimes g = \Pow(P_1 \times Q_1 \| P_2 \times Q_2)$
is defined in terms of the corresponding
product distributions.
\end{definition}

\begin{fact}[Joint Convexity~\cite{dong2022gaussian}]
\label{fact:joint_convexity}
Suppose $\Lambda$ is an indexing set and $f$ is a power function such that
for each $\lambda \in \Lambda$, a pair $(P(\lambda), Q(\lambda))$ of distributions
satisfy $\Pow(P(\lambda) \| Q(\lambda)) \leq f$.

Then, for any distribution $L$ on $\Lambda$,
$\Pow(P(L) \| Q(L)) \leq f$.
\end{fact}

\begin{fact}[Composition Rule of Power Functions \cite{dong2022gaussian}]
\label{fact:composition}
Suppose $P$ and $P'$ are distributions on the sample space $\Omega$.
For each $\omega \in \Omega$,
$Q(\omega)$ and $Q'(\omega)$ are distributions on sample space $\Omega'$.
If $\Pow(P \| P') \leq f$ and $\Pow(Q(\omega) \| Q'(\omega)) \leq g$ for all $\omega \in \Omega$,
then
$
\Pow((P,Q(P)) \| (P',Q'(P')) ) \leq f \otimes g.
$
\end{fact}

\noindent \emph{Distinguishing between Interactive Mechanisms.}
Suppose we wish to compare the behavior of two mechanisms $\mcal{M}_0$ and $\mcal{M}_1$,
to which we only have oracle accesses. An adversary $\mcal{A}: \mcal{R} \times \mcal{U}^* \rightarrow \mcal{U}$
is also an interactive mechanism.
The adversary has some view function $\nu_\algA: \mcal{U} \rightarrow \mcal{U}$
that specifies what information it can gather from the output of the mechanism.
However, later when we instantiate from the abstract mechanism,
we will use the structure of the message space to clarify what the adversary can observe.

We next define the interaction $\algA \leftrightarrow \algM$ between the two mechanisms.
The adversary~$\algA$ repeatedly sends messages to~$\algM$. At
each time step, the adversary determines the next message to be sent
to~$\algM$, based on the history of prior interactions. This process repeats iteratively,
until the mechanism~$\algM$ outputs $\halt$ in some time step.
Formally, it is captured by the following description.

\begin{tcolorbox}
\begin{enumerate}[leftmargin=1pt]
\item The mechanisms
$\algA$ and $\algM$ sample the random seeds $r_\algA$ and $r_\algM$, respectively.

\item At time step $t = 0$, $\algM$ outputs $y_0 \gets \algM(r_\algM) \in \mcal{U}$. 

$\algA$ observes $v_0 := \nu_\algA(y_0)$.

\item At time step $t \geq 1$,

\begin{enumerate}[label=(\alph*)]
\item $\algA$ generates $x_t \gets \algA(r_\algA; v[0..t-1])$ and sends to $\algM$.

\item $\algM$ responds $y_{t} \gets \algM(r_\algM; x[1..t])$;

$\algA$ observes $v_t := \nu_\algA(y_t)$.

\end{enumerate}

The process terminates as soon as either $x_t$ or $y_{t}$ is $\halt$;
otherwise, the process continues for the next time step.

\end{enumerate}
\end{tcolorbox}

\ignore{
Now we define the interaction between an adversary and an interactive mechanism.  Formally, this process is defined as follows:

\begin{definition}\label{def:interaction}
Given an adversary $\algA$ and an interactive mechanism $\algM$, the interaction between $\algA$ and $\algM$, denoted as $(\algA, \algM)$, is the following random process:
\begin{enumerate}
  \item $\algA$, $\algM$ sample the random seeds $r_\algA$, $r_\algM$, respectively.
  \item    For $t=1,2,\dots$:
   \begin{enumerate}
        \item $\algA$ generates the new query based on the history: $q_t := \algA(r_\algA,m_{[t-1]})$.

        \item $\algM$ answers the query from $\algA$: $m_t := \algM(r_\algM, q_{[t]})$.

        \item If $\algA$ or $\algM$ outputs halt, the other does nothing and the loop ends.

    \end{enumerate}
\end{enumerate}
\end{definition}
}

\begin{definition}[Adversary View]
\label{def:view}
Suppose the interaction
 $\algA \leftrightarrow \algM$
terminates at time step~$t$ (which may be a random quantity).
Then, the adversary view
is defined as:

$
\view(\algA \leftrightarrow \algM) := (r_\algA; v[0..t])$.

\end{definition}

\noindent\emph{Terminology (``non-adaptive'').}
In this paper, when we say ``non-adaptive'' for a continual mechanism’s privacy analysis, we mean the mechanism’s input
stream is fixed in advance (equivalently, the environment/adversary that provides inputs does not condition on past outputs).

\begin{definition}[Non-Adaptive Adversary]
An adversary transition function $\mcal{A}: \mcal{R} \times \mcal{U}^* \rightarrow \mcal{U}$
is non-adaptive, if there exists a transition function of the form
$\widehat{\mcal{A}}: \mcal{R} \times \Z \rightarrow \mcal{U}$
such that for any length-$t$ sequence $v[0..t-1]$, $\mcal{A}(r; v[0..t-1]) = \widehat{A}(r; t)$.
\end{definition}

%

%
%

We can next define a distance notion between
two mechanisms 
based on power functions.  Intuitively,
it quantifies the difficulty of distinguishing between
two mechanisms.

\begin{definition}[Power Function between Mechanisms]
\label{defn:power_mech}
Given two interactive mechanisms $\algM_0$ and $\algM_1$,
we overload the notation and denote $\Pow\left(\algM_0 \| \algM_1 \right) \leq f
$, for some power function~$f$, if
for all adversaries $\algA$, 
$\Pow\left(\view(\algA \leftrightarrow \algM_0) \| \view(\algA \leftrightarrow \algM_1)\right) \leq f.
$

If we have the weaker requirement that the inequality holds
for just non-adaptive adversaries,
we denote:
 $\Pow^{\NA}\left(\algM_0 \| \algM_1 \right) \leq f$.

\end{definition}

\noindent \textbf{Simulation of Two Close Interactive
Mechanisms with Close Random Seeds.}
The following result from~\cite{DBLP:conf/stoc/Vadhan023}
states that if two mechanisms are close in the sense
of Definition~\ref{defn:power_mech}, then it is possible to simulate
either one of them exactly with the same mechanism transition function that initially
takes a random seed generated from one of the two appropriate distributions.
We use the finite-space, bounded-termination version of this result.
Under the finite-precision convention above, every bounded-horizon
mechanism considered in this paper falls within its scope.

%

This is a very useful fact in proving privacy composition results,
because instead of interacting with two potentially different mechanisms
in every time step, we can view the adversary as interacting with the
same mechanism transition function but with initial random
seeds drawn from two different distributions.  One useful analogy
is that after a sensitive database is anonymized with appropriate
noise, then it can be queried unlimited number of times
without privacy degradation because of data post-processing.

\begin{fact}[Exact Simulation of Two Close Interactive Mechanisms~\cite{DBLP:conf/stoc/Vadhan023}, Theorem 1.5]
\label{fact:exact_sim}
Suppose $\algM_0$ and $\algM_1: \mcal{R} \times \mcal{U}^* \rightarrow \mcal{U}$
are two interactive mechanisms with bounded termination,
and $f$ is a power function
such that $\Pow(\algM_0 \| \algM_1) \leq f$
in the sense of Definition~\ref{defn:power_mech}.

Then, there exist two distributions
$N_0$ and $N_1$ on some sample space~$\Omega$
and a transition function~$\algP: \Omega \times \mcal{U}^* \rightarrow \mcal{U}$
such that $\Pow(N_0 \| N_1) \leq f$,
and, for all adversaries $\algA$, the following holds
for both $b \in \{0,1\}$:
$\view(\algA \leftrightarrow \algM_b) \equiv \view(\algA \leftrightarrow \algP(N_b)),$

\ignore{
\[
\view(\algA \leftrightarrow \algM_0) \equiv \view(\algA \leftrightarrow \algP(N_0)))
\quad \textrm{and} \quad
\view(\algA \leftrightarrow \algM_1) \equiv \view(\algA \leftrightarrow \algP(N_1)),
\]
}
where $\algP({N})$ is
the interactive mechanism that uses the specified transition function~$\algP$
with a random seed drawn from the distribution~$N$.
\end{fact}

\newcommand{\hop}{\mathsf{hop}}

\section{Modular Composition of Interactive Mechanisms with Hidden Outputs}
\label{sec:modular_comp}

\paragraph{Interactive Mechanism with Hidden Output.}
We next consider a special
form of interactive mechanism
with a transition function of the form
$\algM:\mcal{R} \times \mcal{X}^* \to \mcal{V}\times \mcal{Y}$.
Here, 
we use $\mcal{X}$ to denote the input space.
In each time step,
the mechanism produces $(v, y) \in \mcal{V} \times \mcal{Y}$,
where the adversary can observe~$v$ from
the view space~$\mcal{V}$ and
the output~$y$ from
the output space~$\mcal{Y}$ is potentially hidden from the adversary,
which may be passed to another interactive mechanism (that could however
possibly leak information about $y$).
If an adversary only interacts with $\algM$, then
it can be described by a transition function
with the form
$\algA: \mcal{R} \times \mcal{V}^* \to \mcal{X}$.

\noindent \textbf{Neighboring Input Sequences.}
Similar to the conventional differential privacy,
we need a neighboring notion for input sequences such that
intuitively the behaviors of a mechanism on two neighboring input
sequences should be similar.

\begin{assumption}[Prefix-Closed Neighboring Relation]
\label{assume:ns}
We assume there is a neighboring relation~$\sim_{\mathcal{X}^*}$
(or just $\sim$ for succinctness) that satisfies the property that
if two sequences $x[1..t] \sim x'[1..t']$ are neighboring, then
the following holds:
\begin{compactitem}
\item[(i)] They have the same length, i.e., $t = t'$.

\item[(ii)] All their corresponding prefixes are neighboring,
i.e., for all $1 \leq \tau < t$, $x[1..\tau] \sim x'[1..\tau]$.
\end{compactitem}

\end{assumption}

\noindent \emph{Privacy for Interactive Mechanisms.}
Since the inputs for an interactive mechanism in later time steps
can be influenced by its earlier behavior,
one cannot formally define its privacy just based on neighboring input
sequences.  In the literature, \emph{adaptive} differential privacy~\cite{DBLP:conf/icml/0004RSS23} is based on a construction that we name as \emph{paired simulation}.

\begin{definition}[Paired Simulation]
\label{def:paired_sim}
Suppose an interactive mechanism
is given by the transition function
$\algM: \mcal{R} \times \mcal{X}^* \to \mcal{V} \times \mcal{Y}$
and the random seed distribution $R_\algM$,
where a neighboring relation is defined on $\mcal{X}^*$.

Then, the paired simulation consists
of a canonical pair of interactive mechanisms
$(\algM^\pair_0, \algM^\pair_1)$
such that
for each $b \in \{0, 1\}$,
the transition function has the form $\algM^\pair_b: \mcal{R} \times (\mcal{X}^2)^*
\to \mcal{V} \times \mcal{Y}$.
Moreover, in the paired simulation, an adversary with transition function
 $\mcal{A}: \mcal{R} \times \mcal{V}^* \to \mcal{X}^2$ interacts with
$\algM^\pair_b$ as follows.

\begin{tcolorbox}

The interaction between an adversary $\mcal{A}$
and $\algM^\pair_b$:

\begin{enumerate}[leftmargin=1pt]
\item 
$\algA$ and $\mcal{M}$ sample the hidden random seeds $r_\algA$ and $r_{\algM}$
from their own distributions, respectively.


\item At time step $t = 0$, the simulation calls $(v_0, y_0) \gets \algM(r_{\algM})$.
If $\algM$ is supposed to wait for the first input message,
$(v_0, y_0) = (\bot, \bot)$.

\item At time step $t \geq 1$,
$\mcal{A}$ constructs a pair
$(x^{(0)}_t, x^{(1)}_t) \gets \algA(r_\algA; v[0..t-1]) \in \mcal{X}^2$
and sends to the simulation.

If the sequences $x^{(0)}[1..t]$ and $x^{(1)}[1..t]$
generated so far are not neighboring,
the whole simulation terminates.

Otherwise,
$(v_t, y_t) \gets \algM(r_\algM; x^{(b)}[1..t])$,
and $\algA$ observes~$v_t$.

If either $\algA$ or $\algM$ decides to halt, then the whole simulation
terminates.

\end{enumerate}

The adversarial view is denoted as $\view(\mcal{A} \leftrightarrow \algM^\pair_b)$.
If the simulation does not terminate before
time step~$t$,
then the view up to time~$t$ includes
the random seed $r_\algA$ and $v[0..t]$.

\end{tcolorbox}

\end{definition}

\ignore{
Inspired from differential obliviousness~\cite{DBLP:journals/jacm/ChanCMS22}
and adaptive differential privacy~\cite{DBLP:conf/icml/0004RSS23},
we can define adaptive differential privacy
for interactive mechanisms based on power functions.
}

\begin{definition}[Differentially Private (DP) View]
\label{defn:do}
Given an interactive mechanism~$\algM$,
suppose $(\algM^\pair_0, \algM^\pair_1)$ is the paired
simulation as in Definition~\ref{def:paired_sim}.
For a power function~$f$,
the mechanism~$\algM$ is adaptively $f$-differentially private (DP)
if

$\Pow(\algM^\pair_0 \| \algM^\pair_1) \leq f$,
for all adaptive adversaries in the sense of Definition~\ref{defn:power_mech}.
Similarly, non-adaptively $f$-DP is defined by
$\Pow^\NA(\algM^\pair_0 \| \algM^\pair_1) \leq f$.

\ignore{
Because of Fact~\ref{fact:pow},
it suffices to consider \textbf{deterministic} adversaries
in Definition~\ref{defn:power_mech}.
}
\end{definition}

\noindent \textbf{Multi-Hop Neighboring Notion.}
In the literature, group DP or multi-hop neighbors can be generalized
to sequences readily.

\begin{definition}[$k$-Hop Neighboring Sequences]
\label{defn:hop-neighboring}
For $k=1$, $1$-hop neighboring is the same as that in Assumption~\ref{assume:ns}.
For $k \geq 2$, two sequences $x$ and $x'$ are $k$-hop neighboring
if there exist $x_1, \ldots, x_{k-1}$
such that for $0 \leq i < k$, $x_i$ and $x_{i+1}$ are $1$-hop neighboring,
where $x_0 = x$ and $x_k = x'$.
\end{definition}

We also generalize Definition~\ref{def:paired_sim}
to multi-hop neighboring simulation.

\begin{definition}[$k$-Hop Neighboring Simulation]
\label{defn:hop_sim}
Generalizing the description in Definition~\ref{def:paired_sim},
in each time step~$t$,
the adversary generates a tuple $(x^{(0)}_t, x^{(1)}_t, \ldots, x^{(k)}_t)$
such that for each $0 \leq i < k$,
$x^{(i)}[1..t]$ and $x^{(i+1)}[1..t]$ are neighboring sequences.

Moreover, for $b \in [0..k]$, $\algM^{\hop}_b$ is a simulation
of $\algM$ that takes the $b$-th component $x^{(b)}_t$ at each step~$t$.
\end{definition}

\begin{fact}[Adaptive DP for Multi-Hop Neighboring Sequences]
\label{fact:multi_hop_DP}
Suppose $\algM$ is adaptively $f$-DP as in Definition~\ref{defn:do}.
Then, in a $k$-hop neighboring simulation as in Definition~\ref{defn:hop_sim},
we have:
$\Pow(\algM^\hop_0 \| \algM^\hop_k) \leq f^{\circ k}$,

where $f^{\circ 1} = f$ and $f^{\circ (i+1)} = f^{\circ i} \circ f$.
\end{fact}

\begin{proof}
Fix a $k$-hop adaptive adversary, and let $P_i$ denote the
distribution of its view in simulation $\algM_i^{\hop}$.  For every
$i<k$, projecting each adversarially generated tuple onto coordinates
$(i,i+1)$ gives a valid adaptive paired-simulation adversary.  Hence
adaptive one-hop $f$-DP gives
\[
\Pow(P_i\|P_{i+1})\leq f.
\]
Applying the triangle inequality for power functions repeatedly yields
\[
\Pow(P_0\|P_k)\leq f^{\circ k}.
\]
\end{proof}

%
%

\noindent \textbf{Modular Composition of Interactive Mechanisms.}
As mentioned above, the hidden output of one mechanism
may be fed as the input of another mechanism.
To this end, we define a specific form of composition
that we call \emph{modular composition}.
The adversary only directly supplies an input
to the ``head'' mechanism of the list, where
each intermediate mechanism receives its input from the (hidden)
output from the previous mechanism.

\begin{tcolorbox}
\emph{Modular composition}  is
an interactive mechanism
with a transition function
of the form
$(\mcal{M} \to \mcal{N}):
\mcal{R}_1 \times \mcal{R}_2 \times \mcal{X}^* \to \mcal{V} \times \mcal{W} \times \mcal{Z}$
described as follows.

\begin{enumerate}[leftmargin=1pt]
\item The mechanisms $\algA$, $\mcal{M}$ and $\mcal{N}$ sample the hidden random seeds $r_\algA$, $r_{\algM}$, $r_{\mcal{N}}$
from their own distributions, respectively.

\item At time step $t = 0$,
call the subroutine
$(v_0,y_0) \gets \mcal{M}(r_{\algM})$ and $(w_0, z_0) \gets \mcal{N}(r_\mcal{N})$.

The adversary observes $(v_0, w_0)$.

\item At time step $t \geq 1$,

\begin{enumerate}[label=(\alph*)]
\item the adversary generates $x_t \gets \algA(r_\algA; v[0..t-1], w[0..t-1])$,
and sends it to the head mechanism of the modular composition $(\mcal{M} \to \mcal{N})$.

\item
call the subroutine $(v_t, y_t) \gets \mcal{M}(r_{\algM}; x[1..t])$
and pass the hidden output $y_t$ to the next mechanism $\mcal{N}$;

\item
call the subroutine $(w_t, z_t) \gets \mcal{N}(r_{\mcal{N}}; y[1..t])$;

\item the adversary observes $(v_t, w_t)$.

\end{enumerate}

The whole process terminates as soon as one of
$\mcal{A}$, $\mcal{M}$ and $\mcal{N}$ terminates.
Otherwise, the process continues for the next time step.

\end{enumerate}

If the process terminates at time step~$t$,
the view of the adversary consists
of its random seed $r_\mcal{A}$,
$v[0..t]$ and $w[0..t]$.

\end{tcolorbox}

\begin{definition}[Modular Composition]
\label{defn:list_comp}
Suppose two interactive
mechanisms are given
by the transition functions
$\mcal{M}: \mcal{R}_1 \times \mcal{X}^* \rightarrow \mcal{V} \times \mcal{Y}$
and $\mcal{N}: \mcal{R}_2 \times \mcal{Y}^* \to \mcal{W} \times \mcal{Z}$,
with random seed distributions $R_\mcal{M}$
and $R_\mcal{N}$, respectively.
Then, an adversary with
transition function $\mcal{A}: \mcal{R} \times (\mcal{V} \times \mcal{W})^*
\to \mcal{X}$ interacts
with the modular composition $\mcal{M} \to \mcal{N}$ as above.
\end{definition}

\noindent
\emph{Neighbor-Preserving Paired Simulation.}
To consider composition of differentially 
oblivious mechanisms~\cite{DBLP:journals/jacm/ChanCMS22},
the framework neighbor-preserving differential obliviousness
(NPDO)~\cite{DBLP:conf/eurocrypt/ZhouSCM23,DBLP:conf/innovations/ZhouZCS24}
has been introduced to formalize the condition that if neighboring inputs
are submitted to a mechanism, then the corresponding outputs
would somehow be neighboring for a subsequent mechanism.
To generalize this to interactive mechanisms,
we will need to introduce neighbor-preserving paired simulation.
In addition to neighboring input sequences in $\mcal{X}^*$
for mechanism $\mcal{M}$,
we assume that there is a similar notion of neighboring
sequences in $\mcal{Y}^*$ (satisfying Assumption~\ref{assume:ns})
for mechanism $\mcal{N}$.

\begin{definition}[Neighbor-Preserving Paired Simulation (NPP)]
\label{def:npp}
Given an interactive mechanism
with transition function $\algM: \mcal{R} \times \mcal{X}^* \rightarrow
\mcal{V} \times \mcal{Y}$
and random seed distribution $R_\algM$,
a neighbor-preserving paired simulation (which might not be unique)
is a pair $(\algM^\npp_0, \algM^\npp_1)$
of interactive mechanisms such that the following holds.

\begin{enumerate}[leftmargin=1pt]

\item
There exists possibly another collection $\mcal{R}'$ 
of random seeds and an appropriate distribution $R'_{\mcal{M}}$
such that
for each $b \in \{0,1\}$,
the corresponding mechanism
has a transition function
of the form $\algM^\npp_b: \mcal{R} \times \mcal{R}' \times (\mcal{X}^2)^*
\to \mcal{V} \times \mcal{Y}^2$,
where the random seed distribution on $\mcal{R}$ still follows $R_\algM$,
and the random seed from $\mcal{R}'$ is sampled from~$R'_{\mcal{M}}$.

In other words,
in addition to a random seed in $\mcal{R}$,
an extra random seed in $\mcal{R}'$ may be taken.
Just like paired simulation in Definition~\ref{def:paired_sim},
a sequence of input pairs are also taken.
However, the transition function returns
a tuple in $\mcal{V} \times \mcal{Y}^2$,
i.e., in addition to $v \in \mcal{V}$,
a pair $(y_0, y_1) \in \mcal{Y}^2$ is returned.

\item
For each $b \in \{0,1\}$,
$r \in \mcal{R}$, $r' \in \mcal{R}'$,
$x^{(0)}[1..t]$ and $x^{(1)}[1..t] \in \mcal{X}^*$,

if
$x^{(0)}[1..t]$ and $x^{(1)}[1..t]$ are not neighboring sequences,
then
$\algM^\npp_b(r, r' ; x^{(0)}[1..t], x^{(1)}[1..t]) = \bot$,
and the whole simulation terminates at time step~$t$.

Otherwise,
the pair $(y^{(0)}[1..t],
y^{(1)}[1..t])$ of output sequences
produced in the first $t$ time steps
are neighboring in $\mcal{Y}^*$.

\item
If for some $(v, y_0, y_1) \in \mcal{V} \times \mcal{Y}^2$,
$\algM^\npp_b(r, r' ; x^{(0)}[1..t], x^{(1)}[1..t]) = (v, y_0, y_1)$,
then $\algM(r; x^{(b)}[1..t]) = (v, y_b)$.

In other words,
we require consistency between $\algM^\npp_b$ and $\algM^\pair_b$ (from
Definition~\ref{def:paired_sim}).
Specifically, the projection of $\algM^\npp_b$'s output onto $\algM^\pair_b$'s range must match
the behavior of $\algM^\pair_b$.

However, unlike paired simulation $\algM^\pair_b$ whose
behavior is totally determined by the given $\algM$,
in the tuple $(v, y_0, y_1)$ returned above by $\algM^\npp_b$,
the value $y_{\overline{b}}$ (where $b \neq \overline{b}$) can
depend on all the input parameters in a manner that is not specified by $\algM$.

\item \textbf{Augmented Adversary View.} When an adversary interacts with
$\algM^\npp_b$, the adversarial view consists
of its randomness and the sequence in $(\mcal{V} \times \mcal{Y}^2)^*$ produced
until termination.

\end{enumerate}

\end{definition}

\begin{definition}[NPDP Interactive Mechanisms]
\label{defn:NPDO_mech}
Given a power function~$f$,
an interactive mechanism~$\algM$
is said to be adaptively $f$-neighbor-preserving differentially private (NPDP),
if there exists a neighbor-preserving paired simulation
$(\algM^\npp_0, \algM^\npp_1)$
such that
$\Pow(\algM^\npp_0 \| \algM^\npp_1) \leq f$,
in the sense of Definition~\ref{defn:power_mech}.

\ignore{
Because of Fact~\ref{fact:pow},
it suffices to consider \textbf{deterministic} adversaries
in Definition~\ref{defn:power_mech}.
}
\end{definition}

\begin{remark}[Construction of Neighbor-Preserving Paired Simulation]
Observe that if one wishes
to prove that a mechanism is NPDP as in
Definition~\ref{defn:NPDO_mech},
there is potentially some freedom to construct
an appropriate NPP simulation.

On the contrary, for non-interactive mechanisms (that
only run for one time step),
the recent result on universally closest distribution refinements~\cite{DBLP:conf/innovations/ChanX25}
implies that a ``best possible'' NPP simulation may be used.

Unfortunately, this does not seem to generalize to
interactive mechanisms.  The high level reason is that
there may be no closest distribution refinements
on two random sequences of length~$t+1$
such that its prefix projection would give
a closest distribution refinement for the corresponding
prefixes of length~$t$.  In other words,
if one constructs the ``best possible'' NPP up to time step~$t$,
it may not result in the best NPP if the process enters step~$t+1$.
\end{remark}

\section{Modular Composition Theorem for Adaptively Differentially Private Interactive Mechanisms}
\label{sec:list_composition}

We next describe a modular composition theorem for interactive mechanisms.
At a high level, the theorem states that if the first stage is
\emph{neighbor-preserving} (NPDP) with privacy profile $f$, and the second stage is
(adaptive) DP with profile $g$ on the intermediate stream, then the end-to-end
interactive pipeline is (adaptive) DP with profile $f \otimes g$.

\paragraph{How this theorem relates to existing composition frameworks.}
Our theorem can be viewed as a direct consequence of two existing ideas,
once NPDP is interpreted through \emph{paired simulation}.

First, it is shown~\cite{DBLP:conf/innovations/ZhouZCS24} that
neighbor-preserving notions (including differential obliviousness/NPDO)
admit an equivalent characterization in terms of a paired simulation that
outputs a \emph{refinement pair}.  Informally, in the standard secret-bit
distinguishing game, an adversary provides at each time step a neighboring
pair $(x_t^{(0)},x_t^{(1)})$, and a secret bit $b \in \{0,1\}$ selects which
component is the ``real'' input.  A mechanism is DP if the adversary's views
under $b=0$ and $b=1$ are close.  The refinement-pair view resolves the
usual obstacle in modular composition: even if the real world uses
$x^{(b)}$ and produces a real intermediate output $y^{(b)}$, one also needs
a synthetic neighbor $\overline{y}^{(1-b)}$ so that the next stage can be
analyzed in a paired simulation.  The NPDP guarantee precisely asserts the
existence of such a paired simulator that, for each $b$, produces a pair
$(y^{(0)},y^{(1)})$ whose $b$-marginal matches the real execution, and such
that the two resulting pair-distributions (for $b=0$ vs.\ $b=1$) are close.

Second, a
general framework for adaptive privacy of continual mechanisms~\cite{Henzinger2026Continual}
is developed to prove
concurrent composition theorems under interactive adversaries.
In particular, once a mechanism is expressed in the standard secret-bit game
(and the controller only interacts with downstream submechanisms through
their privatized outputs), privacy of the overall system follows from
standard interactive composition and post-processing.
Under the paired-simulation perspective, our two-stage pipeline
$\algM \rightarrow \mcal{N}$ fits this template: the NPDP stage can be viewed
as a mechanism that, on secret bit $b$, outputs a paired intermediate stream
$(y^{(0)}_t,y^{(1)}_t)$ (together with visible outputs $v_t$), and the DP stage
$\mcal{N}$ is then run on $y^{(b)}$.
Thus, one may reduce modular composition to ordinary interactive composition:
compose the lifted NPDP mechanism (as DP-on-the-bit with profile $f$) with the
second-stage DP mechanism (profile $g$), and finally apply post-processing to
remove the extra synthetic transcript.

While this reduction gives a useful conceptual guide and could be formalized
by instantiating the conditions of~\cite{Henzinger2026Continual}
together with the refinement-pair characterization
of~\cite{DBLP:conf/innovations/ZhouZCS24}, we include below a self-contained
proof tailored to our setting and our power-function accounting.
In particular, a more general concurrent setting (multiple
submechanisms with interleaving calls) is studied in~\cite{Henzinger2026Continual}, whereas we only need the serial
modular composition used by our buffering--aggregation pipeline.  Providing
a direct proof also makes explicit how the privacy profile tensor $f \otimes g$
arises in our notation, and it cleanly interfaces with the certification
arguments in other sections.


\ignore{
\begin{definition}
Suppose $f$ and $g$ are power functions
that are attained by the corresponding distributions in
$\Pow(P_1 \| P_2) = f$ and $g = \Pow(Q_1 \| Q_2)$.
Then,
their corresponding tensor
product $f \otimes g = \Pow(P_1 \times Q_1 \| P_2 \times Q_2)$
is defined in terms the corresponding
product distributions.
\end{definition}
}

\begin{theorem}
\label{thm:DP-comp}
  Suppose an interactive mechanism with transition function $\algM: \mcal{R}_1 \times \mcal{X}^* \to \mcal{V} \times \mcal{Y}$ is adaptively $f$-NPDP (as in Definition~\ref{defn:NPDO_mech}), and
	another one with
	$\mcal{N}: \mcal{R}_2 \times \mcal{Y}^* \to \mcal{W}\times \mcal{Z}$ is adaptively $g$-DP (as in Definition~\ref{defn:do}). Then, 
	the modular composition (as in Definition~\ref{defn:list_comp}) $\algM \rightarrow \mcal{N}: \mcal{R}_1 \times \mcal{R}_2 \times \mcal{X}^* \to \mcal{V} \times \mcal{W}\times  \mcal{Z}$ is adaptively $f \otimes g$-DP.
\end{theorem}

\begin{tcolorbox}
The augmented adversary $\widehat{\mcal{A}}$ interacts
with $\widehat{\mcal{L}}_b$, for each $b \in \{0,1\}$, as follows.

\begin{enumerate}[leftmargin=1pt]
\item The mechanisms 
$\algA$, $(\mcal{M}^\npp_0, \mcal{M}^\npp_1)$ and $(\mcal{N}^\pair_0,
\mcal{N}^\pair_1)$, sample the random seeds $r_\algA$, $(r_{\algM}, r'_{\algM})$, $r_{\mcal{N}}$
from the appropriate distributions, respectively.

\item At time step $t = 0$,
call the subroutines
$(v_0, y^{(0)}_0, y^{(1)}_0 ) \gets \mcal{M}^\npp_b(r_{\algM}, r'_{\algM})$ and $(w_0, z_0) \gets \mcal{N}^\pair_b(r_\mcal{N})$.

$\widehat{\mcal{A}}$ observes $(v_0, w_0)$.

\item At time step $t \geq 1$,

\begin{enumerate}[label=(\alph*)]
\item $\widehat{\mcal{A}}$ generates $(x^{(0)}_t, x^{(1)}_t ) \gets \algA(r_\algA; v[0..t-1], w[0..t-1])$,
and sends it to $\widehat{\mcal{L}}_b$.

\item
$(v_t, y^{(0)}_t, y^{(1)}_t) \gets 
\mcal{M}^\npp_b(r_{\algM}, r'_{\algM}; x^{(0)}[1..t], x^{(1)}[1..t])$
and pass the hidden output pair $(y^{(0)}_t, y^{(1)}_t)$ to the next mechanism $\mcal{N}^\pair_b$;

\item
$(w_t, z_t) \gets 
\mcal{N}^\pair_b(r_{\mcal{N}}; y^{(0)}[1..t], y^{(1)}[1..t])$;

\item $\widehat{\mcal{A}}$  observes $(v_t, w_t)$.

\end{enumerate}

The whole process terminates as soon as one of
$\mcal{A}$, $\mcal{M}^\npp_b$ and $\mcal{N}^\pair_b$ terminates.
Otherwise, the process continues for the next time step.

\end{enumerate}

If the process terminates at time step~$t$,
the view of the adversary consists
of its random seed $r_\mcal{A}$,
$v[0..t]$ and $w[0..t]$.

\end{tcolorbox}

\begin{proof}
We denote the resulting interactive
mechanism from the modular composition
by $\mcal{L} := \algM \rightarrow \mcal{N}$.
In view of Definition~\ref{defn:do},
we consider the paired simulation
$(\mcal{L}^\pair_0, \mcal{L}^\pair_1)$
as described in Definition~\ref{def:paired_sim}.
Specifically,
our goal is to show that for any adversary $\mcal{A}$,
the two distributions $\view(\mcal{A} \leftrightarrow \mcal{L}^\pair_0)$
and $\view(\mcal{A} \leftrightarrow \mcal{L}^\pair_1)$
are ``close'' in the sense quantified by the power function
$f \otimes g$.

Our proof strategy is to construct 
an augmented pair $(\widehat{\mcal{L}}_0, \widehat{\mcal{L}}_1)$
of interactive mechanisms
and an augmented adversary $\widehat{\mcal{A}}$
whose views
have the same distributions as those from
$(\mcal{L}^\pair_0, \mcal{L}^\pair_1)$ and $\mcal{A}$.

We next use Fact~\ref{fact:exact_sim} to show that 
instead of considering 
the pair $(\mcal{L}^\pair_0, \mcal{L}^\pair_1)$
with the same random seed distribution,
we can equivalently simulate the pair with
some common transition function, but using two
different random seed distributions.
By the data processing inequality,
it suffices to analyze the alternative simulation
and the power function of the two random seed distributions.
Because $\mcal{M}$ is $f$-NPDP,
let $(\mcal{M}^\npp_0, \mcal{M}^\npp_1)$ be the corresponding
NPP simulation as guaranteed by Definition~\ref{defn:NPDO_mech}
such that 
$\Pow(\mcal{M}^\npp_0 \| \mcal{M}^\npp_1) \leq f$.
By Fact~\ref{fact:exact_sim},
there exist some interactive mechanism $\mcal{P}$
and two random seed distributions $F_0$ and $F_1$
such that $\Pow(F_0 \| F_1) \leq f$
and for each $b \in \{0 ,1\}$,
$\mcal{P}(F_b)$ is an exact simulation of $\mcal{M}^\npp_b$
(with the appropriate random seed distribution
$R_{\mcal{M}} \times R'_{\mcal{M}}$)
for any adversary.
Specifically, for each~$b$,
$\mcal{M}^\npp_b(R_{\mcal{M}} \times R'_{\mcal{M}})$
and $\mcal{P}(F_b)$ are equivalent.

Similarly,
since $\mcal{N}$ is $g$-DP,
the paired simulation $(\mcal{N}^\pair_0, \mcal{N}^\pair_1)$
satisfies $\Pow(\mcal{N}^\pair_0 \| \mcal{N}^\pair_1 ) \leq g$.
Again, by
Fact~\ref{fact:exact_sim},
there exist some interactive mechanism $\mcal{Q}$
and two random seed distributions $G_0$ and $G_1$
such that $\Pow(G_0 \| G_1) \leq g$
and for each $b \in \{0 ,1\}$,
$\mcal{Q}(G_b)$ is an exact simulation of $\mcal{N}^\pair_b$
(with the appropriate random seed distribution $R_\mcal{N}$)
for any adversary.
Specifically, for each~$b$,
$\mcal{N}^\pair_b(R_{\mcal{N}})$
and $\mcal{Q}(G_b)$ are equivalent.

Given an adversary $\mcal{A}$ that
interacts with $\mcal{L}^\pair_b$ for $b \in \{0,1\}$,
we describe how an augmented adversary $\widehat{\mcal{A}}$ interacts
with $\widehat{\mcal{L}}_b$ in the above colored box.
Because we follow exactly the same structure as in Definition~\ref{defn:list_comp},
after comparing the two descriptions line-by-line,
it follows that
$\view(\mcal{A} \leftrightarrow \mcal{L}^\pair_b)$
has the same distribution as $\view(\widehat{\mcal{A}} 
\leftrightarrow \widehat{\mcal{L}}_b)$.

As aforementioned,
we use Fact~\ref{fact:exact_sim},
for each $b \in \{0, 1\}$,
we can equivalently replace $\mcal{M}^\npp_b(R_{\algM}, R'_{\algM})$
with $\mcal{P}(F_b)$, and 
$\mcal{N}^\pair_b(R_{\mcal{N}})$ with 
$\mcal{Q}(G_b)$ in the above description.

Hence, the first inequality below follows from the data processing inequality:

\begin{align*}
\Pow\!\big(\view(\mcal{A} \leftrightarrow \mcal{L}^\pair_0)\,\big\|\,\view(\mcal{A} \leftrightarrow \mcal{L}^\pair_1)\big)
&= \Pow\!\big(\view(\widehat{\mcal{A}} \leftrightarrow \widehat{\mcal{L}}_0)\,\big\|\,\view(\widehat{\mcal{A}} \leftrightarrow \widehat{\mcal{L}}_1)\big)\\
&\le \Pow(F_0 \times G_0 \,\|\, F_1 \times G_1)\\
&\le f \otimes g,
\end{align*}

where the second inequality follows from the data processing inequality and the exact simulations,
and the last inequality follows from Fact~\ref{fact:composition}.
\end{proof}

\section{RandBin: Achieving Privacy for Edit-Style Neighboring Inputs}
\label{sec:randbin}

Continual-release mechanisms are usually analyzed under Hamming-style
neighboring streams. Given a collection~$\mcal{O}$ of objects, an infinite
input stream is a map $x:\Z_{\geq1}\rightarrow\mcal{O}$. Two streams
$x\sim_H x'$ are Hamming-neighboring if they differ in at most one index.
In our setting, however, participation itself is sensitive, so one stream
may instead be obtained from the other by inserting or deleting one item.

\begin{definition}[Neighboring in the Edit Sense]
\label{defn:neighbor_edit}
Given a stream~$x$ and $t\in\Z_{\geq1}$, let
$\mathsf{Delete}(x;t)$ be the stream~$x'$ defined by
\[
x'[\tau]
=
\begin{cases}
x[\tau], & \tau<t,\\
x[\tau+1], & \tau\geq t.
\end{cases}
\]
Two streams $x\simE x'$ are neighboring in the edit sense if one is
obtained from the other by one such deletion. Two finite sequences of the
same length are edit-neighboring if they are prefixes of two
edit-neighboring infinite streams.
\end{definition}

\subsection{Motivating Example: Private Prefix Sums Under Edit Adjacency}
\label{sec:example_prefixsum}

For an integer stream $\sigma$, the prefix-sum workload releases estimates
of $\sum_{i=1}^t\sigma[i]$ for every~$t$. Standard tree-based mechanisms
are designed for Hamming neighbors, where one bounded coordinate changes
between the two streams~\cite{DBLP:conf/stoc/DworkNPR10,DBLP:journals/tissec/ChanSS11}.
An edit can instead shift every subsequent coordinate. For example,
deleting the first element of
\[
(1,0,1,0,\ldots)
\]
produces $(0,1,0,1,\ldots)$, so the two streams differ at every position.
Thus the usual one-coordinate sensitivity analysis does not apply
directly.

The role of $\randbin$ is to bridge these neighboring notions.
It buffers arriving items and occasionally emits a random-sized prefix of
the buffer as one concrete bin. Consider the neighboring streams
\[
\sigma=[a,b,c,d,e,f],
\qquad
\sigma'=[b,c,d,e,f,g].
\]
One possible pair of binnings is
\begin{align*}
\sigma
&:\ [\bot,\bot,\{a,b\},\bot,\bot,\{c,d,e\}],\\
\sigma'
&:\ [\bot,\bot,\{b\},\bot,\bot,\{c,d,e\}].
\end{align*}
The emitted sequences now differ in only one concrete bin. The formal
construction below makes the release schedules identical and confines the
effect of one edit to at most two adjacent concrete bins, while ensuring
that the two paired executions have close distributions.

\subsection{Technical Description of $\randbin$}

We first define the structured Hamming-style relation used at the output
interface.

\begin{definition}[Neighboring Bin Sequences]
\label{defn:neighbor_bin}
For the parameter~$U$ fixed below, let $\mcal{B}$ be the collection of
non-empty concrete bins, each containing an array of at most $2U+1$
objects. A bin sequence of length~$T$ is a map
\[
B:[1..T]\rightarrow\mcal{B}\cup\{\bot\},
\]
where $\bot$ denotes that no concrete bin is emitted. Two concrete bins at
times $i<j$ are adjacent if all outputs strictly between them are~$\bot$.

Two bin sequences $B\simB B'$ are neighboring if they have the same
concrete-bin emission times and differ in at most two positions. More
precisely:
\begin{enumerate}
\item If they differ in exactly one position, that position is the last
concrete bin in both sequences, and the two bins have the form
\[
[\sigma_1,w,\sigma_2]
\qquad\text{and}\qquad
[\sigma_1,\sigma_2,u].
\]
\item If they differ in exactly two positions, these positions contain
adjacent concrete bins in both sequences, and the two pairs have the form
\begin{align*}
[\sigma_1,w,\sigma_2],\ \bot,\ldots,\bot,\ [u,\sigma_3],\\
[\sigma_1,\sigma_2,u],\ \bot,\ldots,\bot,\ [\sigma_3],
\end{align*}
for items $w,u$ and possibly empty arrays
$\sigma_1,\sigma_2,\sigma_3$.
\end{enumerate}
Identical bin sequences are also neighboring. In every case, concatenating
the concrete bins gives finite prefixes of edit-neighboring item streams.
\end{definition}

Thus $\simB$ is a structured distance-at-most-two relation with respect to
ordinary per-bin Hamming adjacency. If a downstream mechanism is analyzed
only for one-coordinate Hamming neighbors, its distance-two, or group-DP,
profile must be used when invoking the modular composition theorem.

\begin{definition}[Symmetric Geometric Distribution]
For $\alpha>1$, the symmetric geometric distribution
$\text{Geom}(\alpha)$ on~$\Z$ has probability mass function
\[
\Pr{\text{Geom}(\alpha)=k}
=
\frac{\alpha-1}{\alpha+1}\alpha^{-|k|}.
\]
\end{definition}

\begin{definition}[Truncated Geometric Distribution]
\label{def:truncated_geom}
Let $\varepsilon>0$ and $\delta\in(0,1)$. Let $U$ be the smallest even
positive integer such that
\[
\Pr{\left|\text{Geom}(e^\varepsilon)\right|\geq U/2}\leq\delta.
\]
In particular,
\[
U=O\left(\frac{1}{\varepsilon}\log\frac{1}{\delta}\right).
\]
The truncated geometric distribution $G(\varepsilon,\delta)$ is obtained
by sampling $R\sim\text{Geom}(e^\varepsilon)$ and returning
\[
\operatorname{clip}(R,[-U/2,U/2])
:=
\min\{\max\{R,-U/2\},U/2\}.
\]
\end{definition}

Algorithm~\ref{alg:randbin} uses the prefix-sum mechanism
$\presum_{\varepsilon_2,\delta_2}$ from
Theorem~\ref{th:prefix_sum}, instantiated with $[a..b]=[U..2U]$.
Write
\[
E[i]:=E_{\varepsilon_2,\delta_2}(i).
\]
The theorem gives all properties needed here: the mechanism is adaptively
$(\varepsilon_2,\delta_2)$-DP for unit-Hamming integer streams; on every
actual load stream in $[U..2U]$, it satisfies
\[
\left|S[i]-\sum_{j=1}^i C[j]\right|\leq E[i],
\qquad
S[i]-S[i-1]\in[U..2U]
\]
simultaneously for every~$i$ and every realization of its randomness.

\begin{algorithm}[H]
\caption{Interactive $\randbin$ Mechanism}
\label{alg:randbin}
\begin{algorithmic}[1]
\STATE \textbf{Input:} $\varepsilon>0$, $\delta\in(0,1)$; a stream
$x:\Z_{\geq1}\rightarrow\mcal{O}$
\STATE \textbf{Output:} A bin sequence
$\Z_{\geq1}\rightarrow\mcal{B}\cup\{\bot\}$; the visible output reveals
only whether a concrete bin or $\bot$ is emitted
\STATE Set
$\varepsilon_1=\varepsilon_2=\varepsilon/2$ and
$\delta_1=\delta_2=\delta/2$
\STATE Let $U$ be the value fixed by
Definition~\ref{def:truncated_geom}
\STATE Let
\[
\mcal{G}:=\frac{3U}{2}+G(\varepsilon_1,\delta_1),
\]
which has support $[U..2U]$
\STATE Independently initialize
$\presum_{\varepsilon_2,\delta_2}$ from
Theorem~\ref{th:prefix_sum} for the interval $[U..2U]$, and set
$E[i]:=E_{\varepsilon_2,\delta_2}(i)$
\STATE Initialize an empty FIFO queue $\buf\gets\emptyset$
\STATE Set
$\ind\gets0$, $C[0]\gets0$, $S[0]\gets0$, $E[0]\gets0$, and
$\bin[0]\gets\bot$
\FOR{each step $t\geq1$ at which $x[t]$ is received}
  \STATE Append $x[t]$ to $\buf$
  \IF{$t<S[\ind]+E[\ind]+2U$} \algline{ln:compare}
    \STATE Output $\bot$
  \ELSE
    \STATE If $C[\ind]=0$, set $\bin[\ind]\gets\bot$; otherwise, remove
    the first $C[\ind]$ items from $\buf$ to form $\bin[\ind]$
    \algline{ln:remove_b}
    \STATE Output $\bin[\ind]$
    \STATE $\ind\gets\ind+1$
    \STATE Sample $C[\ind]\sim\mcal{G}$ independently
    \STATE
    $S[\ind]\gets\presum_{\varepsilon_2,\delta_2}(C[\ind])$
  \ENDIF
\ENDFOR
\end{algorithmic}
\end{algorithm}

\begin{theorem}[Utility Guarantee of $\randbin$]
\label{thm:randbin_utility}
At the end of every step~$t$, for every realization of the randomness,
the concatenation of the concrete bins emitted so far is a prefix of the
input stream, and the number of items remaining in $\buf$ is
\[
O\left(
1+
\frac{1}{\varepsilon}\log(t+1)
\left(
\log(t+1)+\log\frac{1}{\delta}
\right)
\right).
\]
Equivalently, the emitted prefix has length at least $t$ minus this
quantity.
\end{theorem}

\begin{proof}
Let
\[
A_i:=\sum_{j=1}^i C[j],
\qquad A_0:=0.
\]
Suppose bin~$i\geq1$ is emitted at step~$t$. The emission condition and
the always-valid error bound imply
\[
t
\geq S[i]+E[i]+2U
\geq A_i+2U.
\]
Immediately before removing bin~$i$, the queue therefore contains
\[
t-A_{i-1}\geq C[i]+2U
\]
items, so underflow is impossible. Since the queue is FIFO, the emitted
items always form a prefix of the input stream.

At the end of every step, the current index~$\ind$ satisfies
\begin{equation}
\label{eq:randbin_threshold_invariant}
t<S[\ind]+E[\ind]+2U.
\end{equation}
This is immediate if no bin is emitted. If a bin is emitted, the previous
threshold is crossed for the first time at the current integer step, while
\[
S[\ind+1]-S[\ind]\geq U\geq1
\quad\text{and}\quad
E[\ind+1]\geq E[\ind],
\]
so the next threshold is strictly larger than~$t$.

If $\ind\geq1$, the number of buffered items at the end of step~$t$ is
$t-A_{\ind-1}$. Using
\eqref{eq:randbin_threshold_invariant}, consistency, and the error bound,
\begin{align*}
t-A_{\ind-1}
&<
S[\ind]+E[\ind]+2U
-
\bigl(S[\ind-1]-E[\ind-1]\bigr)\\
&\leq
4U+E[\ind]+E[\ind-1].
\end{align*}
For $\ind=0$, it is simply $t<2U$. Since $\ind\leq t$,
Definition~\ref{def:truncated_geom} and
Theorem~\ref{th:prefix_sum} give the claimed bound.
\end{proof}

\subsection{Adaptive Privacy Analysis}

The visible release schedule of $\randbin$ depends only on its sampled
loads and the prefix-sum scheduler, not on the item values. Nevertheless,
the concrete bins are hidden outputs passed to the downstream mechanism.
We must therefore control the joint distribution of the visible schedule
and these hidden outputs, as required by adaptive NPDP in
Definition~\ref{defn:NPDO_mech}.

\paragraph{The NPP simulation.}
Recall from Definition~\ref{def:npp} that an NPP simulator receives paired
edit-neighboring input prefixes $(x^{(0)},x^{(1)})$. In branch~$b$, it must
simulate the real execution on $x^{(b)}$ and simultaneously construct a
synthetic output sequence for $x^{(1-b)}$ so that the two output sequences
are $\simB$-neighboring. The construction below makes explicit the paired
refinement underlying the static RandBin analysis
in~\cite{DBLP:conf/eurocrypt/ZhouSCM23}.

As a standard bookkeeping convention, we assume that each logical item
appears at most once and carries a unique tag, preserved across the two
neighboring streams and ignored by downstream computation. Thus, if the
two input streams first differ at step~$\tau$, then at the beginning of
step~$\tau+1$ the simulator can determine which stream contains the
unmatched item~$w$. This one-step delay is harmless: the $2U$ slack
ensures that $w$ cannot be emitted at step~$\tau$.

\begin{definition}[NPP Simulation for $\randbin$]
\label{defn:randbin_npp}
For $b\in\{0,1\}$, $\randbin^{\npp}_b$ runs
Algorithm~\ref{alg:randbin} on $x^{(b)}$, producing the visible release
schedule and the real bin sequence~$y^{(b)}$. It uses the same release
times to construct the synthetic sequence~$y^{(1-b)}$.

Until the unmatched item affects an emitted bin, the two bin sequences are
identical. Suppose first that the real branch is the longer input stream,
and let the real adjacent concrete bins containing the unmatched item~$w$
and the following item~$u$ be
\[
[\sigma_1,w,\sigma_2],
\ \bot,\ldots,\bot,\
[u,\sigma_3].
\]
The simulator outputs in the synthetic sequence
\[
[\sigma_1,\sigma_2,u],
\ \bot,\ldots,\bot,\
[\sigma_3].
\]
If the second concrete bin has not yet been emitted, only the first
replacement is made, and the second is deferred. If the real branch is
the shorter stream, the simulator applies the inverse transformation.
The construction for $b=1$ is symmetric.

Consequently, the two output sequences have identical release times and
are neighboring according to Definition~\ref{defn:neighbor_bin}.
\end{definition}

The construction is online. Indeed, suppose the first affected bin is
bin~$J$. At its emission time~$t$,
\[
t\geq S[J]+E[J]+2U\geq A_J+2U.
\]
Thus, after removing its $C[J]$ items, at least one further item~$u$ is
already available in the buffer.

We now isolate the two privacy components: the shifted fresh load and the
private prefix-sum scheduler.

\begin{lemma}[Adaptive Privacy of $\randbin$] 
\label{lemma:randbin_NPDO}
Let $Z\sim\mcal{G}$. Suppose that, for some power function
$f_{\mathrm{load}}$,
\begin{equation}
\label{eq:randbin_load_profile}
\Pow\left(
Z
\middle\|
Z-\theta
\right)
\leq f_{\mathrm{load}}
\qquad
\text{for both }\theta\in\{-1,+1\}.
\end{equation}
Suppose also that the prefix-sum scheduler is adaptively
$f_{\mathrm{sum}}$-DP for unit-Hamming integer streams and uses randomness
independent of the load samples. Then the NPP simulation in
Definition~\ref{defn:randbin_npp} satisfies
\[
\Pow\left(
\randbin^{\npp}_0
\middle\|
\randbin^{\npp}_1
\right)
\leq
f_{\mathrm{load}}\otimes f_{\mathrm{sum}}.
\]
Hence $\randbin$ is adaptively
$(f_{\mathrm{load}}\otimes f_{\mathrm{sum}})$-NPDP.
\end{lemma}

\begin{proof}
Fix a finite horizon~$T$. By joint convexity
(Fact~\ref{fact:joint_convexity}), it suffices to fix a deterministic
adaptive adversary~$\algA$. We analyze the stronger view that also reveals
the paired bin lengths and every prefix-sum output~$S[k]$ when it is
generated. The original NPP transcript is a post-processing of this view.

If no edit occurs, the two branches are identical. Otherwise, let
$\tau$ be the first step at which the input prefixes differ, and let~$i$
be the current bin index immediately before the item at step~$\tau$ is
appended. If $i=0$, this item cannot enter bin~$0$ because $C[0]=0$.
If $i\geq1$, bin~$i-1$ was emitted at some step $r<\tau$. At the beginning
of step~$\tau$, before appending the new item, the buffer therefore
contains at least
\[
r-A_{i-1}
\geq
S[i-1]+E[i-1]+2U-A_{i-1}
\geq2U
\]
older items. Since $C[i]\leq2U$, the new item cannot enter bin~$i$.

Let $J$ be the first affected concrete-bin index. The preceding argument
gives $J\geq i+1$. The simulator preserves the length of bin~$J$ by
replacing~$w$ with~$u$; hence the paired length streams
$D^{(0)},D^{(1)}$ differ only at
\[
K:=J+1,
\]
where their values differ by exactly~$1$. The deletion direction is known
at the beginning of step~$\tau+1$, whereas at most one bin can be emitted
in any step. Thus bin~$J$ cannot be emitted before the direction is known.
Moreover, Algorithm~\ref{alg:randbin} samples $C[K]$ only after bin~$J$
has been emitted. Consequently, both $K$ and the sign
\[
\theta:=D^{(1)}[K]-D^{(0)}[K]\in\{-1,+1\}
\]
are determined before $C[K]$ is sampled.

Let $\mcal{F}_k$ be the complete augmented history immediately before
$C[k]$ is sampled. Since the loads are fresh and independent,
\begin{equation}
\label{eq:randbin_fresh_load}
\operatorname{Law}(C[k]\mid\mcal{F}_k)=\mcal{G}.
\end{equation}
For $k\neq K$, the paired lengths satisfy
\[
(D^{(0)}[k],D^{(1)}[k])=(C[k],C[k])
\]
in both branches. At coordinate~$K$, they satisfy
\[
(D^{(0)}[K],D^{(1)}[K])
=
\begin{cases}
(C[K],C[K]+\theta), & \text{in branch }0,\\
(C[K]-\theta,C[K]), & \text{in branch }1.
\end{cases}
\]
Thus, conditioned on every reachable pre-$K$ history, the first coordinate
has law~$Z$ in branch~$0$ and law~$Z-\theta$ in branch~$1$. In either
branch, the second coordinate is obtained from the first by adding
$\theta$. Hence the paired release is obtained from the corresponding
scalar experiment through the same injective map
\[
a\longmapsto(a,a+\theta),
\]
whose inverse on its image is projection onto the first coordinate.
Its conditional power function is therefore exactly
\[
\Pow(Z\|Z-\theta)
\leq f_{\mathrm{load}}
\]
by \eqref{eq:randbin_load_profile}.

All earlier and later load kernels are common to the two branches. By the
composition rule for power functions, joint convexity, and data processing,
the interactive paired-length generator is adaptively
$f_{\mathrm{load}}$-NPDP. This conclusion remains valid although $K$ is
selected from previous scheduler outputs, because the conditional bound
holds at every reachable pre-$K$ history.

The streams $D^{(0)}$ and $D^{(1)}$ are unit-Hamming neighbors. Their
synthetic coordinates may lie in $[U-1..2U+1]$, but
Theorem~\ref{th:prefix_sum} defines the scheduler and guarantees adaptive
privacy on all integer streams; the interval $[U..2U]$ is needed only for
the execution-wise consistency guarantee on the actual load stream.
Applying the modular composition theorem
(Theorem~\ref{thm:DP-comp}) to the paired-length generator and the
adaptive prefix-sum scheduler gives the joint profile
\[
f_{\mathrm{load}}\otimes f_{\mathrm{sum}}.
\]

Finally, the input pairs, paired lengths, and scheduler outputs determine
the complete NPP transcript by the FIFO rule: the scheduler outputs
determine the emission times, and the paired lengths determine the two bin
contents. As shown above, the $2U$ slack makes this reconstruction causal.
The claimed bound follows by data processing. Since $T$ and $\algA$ were
arbitrary, the result holds for the unbounded interaction.
\end{proof}

\begin{corollary}
\label{cor:adap_randbin}
The interactive mechanism $\randbin$ in
Algorithm~\ref{alg:randbin} is adaptively
$(\varepsilon,\delta)$-NPDP with input neighboring relation~$\simE$ and
output neighboring relation~$\simB$.
\end{corollary}

\begin{proof}
Let $Z\sim\mcal{G}$ and fix $\theta\in\{-1,+1\}$. Away from the atoms
inherited from the two clipped tails, the likelihood ratio between $Z$
and $Z-\theta$ is at most $e^{\varepsilon_1}$. The exceptional atoms
under $Z$ have total mass at most
\[
\Pr{\left|\text{Geom}(e^{\varepsilon_1})\right|\geq U/2}
\leq\delta_1.
\]
Hence
\[
\Pow(Z\|Z-\theta)
\leq
\DP_{\varepsilon_1,\delta_1}
\qquad
\text{for both }\theta\in\{-1,+1\}.
\]
By Theorem~\ref{th:prefix_sum}, the scheduler is adaptively
$(\varepsilon_2,\delta_2)$-DP for unit-Hamming integer streams.
Lemma~\ref{lemma:randbin_NPDO} and basic composition now give
\begin{align*}
\Pow\left(
\randbin^{\npp}_0
\middle\|
\randbin^{\npp}_1
\right)
&\leq
\DP_{\varepsilon_1,\delta_1}
\otimes
\DP_{\varepsilon_2,\delta_2}\\
&\leq
\DP_{\varepsilon_1+\varepsilon_2,\,
\delta_1+\delta_2}
=
\DP_{\varepsilon,\delta}.
\end{align*}
\end{proof}

\section{Adaptive Safety via Quantitative Decomposability}
\label{sec:ind_decomp}

This section proves the adaptive-safety certification used in the modular
pipeline.  The goal is to identify structural conditions under which an
ordinary non-adaptive privacy analysis of a continual mechanism already
implies privacy against adaptive adversaries.  The result applies to
Hamming-style neighboring input streams and is stated in the power-function
language used throughout the appendix.

The first condition is a fresh-randomness, prefix-causality condition.  At
round~$t$, the mechanism may depend on the entire input prefix $x[1..t]$, but
the only randomness used for the round-$t$ output is a fresh seed
$\omega_t$ independent of all other rounds.  This rules out hidden
cross-round correlations introduced by reused noise, while still allowing the
release at time~$t$ to be an arbitrary deterministic function of the input
prefix.

\begin{definition}[Independently Decomposable Mechanisms]
\label{defn:indep_decomp}
An (abstract) interactive mechanism with
transition function 
$\algM: \mcal{R} \times \mcal{U}^*\rightarrow \mcal{U}$
is independently decomposable if the following properties hold.

\begin{compactitem}

\item Each random seed $r \in \mcal{R}$
takes the form $r = (\omega_1, \omega_2, \ldots) \in \Omega^{\N}$.
Moreover, the mechanism
samples $\omega_t$ independently from some distribution~$R_t$ in time step~$t$.

\item  For each~$t \geq 1$,
there exists a deterministic function $\algM_t$
such that the output $y_t \gets \algM_t(\omega_t; x[1..t]) \in \mcal{U}$
depends only on $x[1..t]$ and the randomness~$\omega_t$ sampled in step~$t$.

\end{compactitem}

\end{definition}

\paragraph{Quantitatively Decomposable Mechanisms.}
Independent decomposability alone records where the randomness enters, but
does not quantify how much a change in the input prefix can affect a
one-round release.  We therefore add a quantitative condition.  The condition
says that, for each round~$t$, the power function between the two possible
round-$t$ output laws is determined exactly by the distance between two
aggregates of selected input coordinates.  The invariant-pair condition below
ensures that common aggregation context can be ignored: only the part of the
aggregate containing the Hamming discrepancy contributes to the distance.

\begin{definition}[Invariant Pair]
\label{defn:inv_pair}
Let $d:\mcal{U}\times\mcal{U}\rightarrow\R_{\geq 0}$ be a metric, and let $(\mcal{U},\oplus,0_{\oplus})$ be a monoid. 
We say that $(d, \oplus)$ is an invariant pair 
(under common aggregation context) if for all $a,b,u,u'\in\mcal{U}$,
$$
d(a\oplus u\oplus b, a\oplus u'\oplus b)=d(u,u').
$$
\end{definition}

\noindent \textbf{Notation.}
For an index set $I=\{\tau_1<\tau_2<\cdots<\tau_k\}\subseteq[1..t]$ and an input history $x[1..t]\in\mcal{U}^t$, define
$$
\bigoplus_{\tau\in I} x[\tau]
:=
x[\tau_1]\oplus x[\tau_2]\oplus\cdots\oplus x[\tau_k].
$$
For $I=\emptyset$, we use the convention
$\bigoplus_{\tau\in I}x[\tau]:=0_{\oplus}$.
Since $\oplus$ is associative and the elements are aggregated in increasing index order, this expression is well-defined without requiring commutativity.

\begin{definition}[Quantitatively Decomposable Mechanism]
\label{defn:quant_decomp}
An interactive mechanism
$\algM:\mcal{R}\times\mcal{U}^*\rightarrow\mcal{U}$
is \emph{quantitatively decomposable} with respect to an invariant pair
$(d,\oplus)$ in the sense of Definition~\ref{defn:inv_pair}
if it is independently decomposable in the sense of
Definition~\ref{defn:indep_decomp}, and there exist a sequence of index sets
$\{I_t\}_{t\geq 1}$ with $I_t\subseteq[1..t]$ and, for every
$t\geq 1$ and $\theta\geq 0$, a power function
$g_t(\theta):[0,1]\rightarrow[0,1]$ such that the following hold.

\begin{compactitem}
\item For every $t\geq 1$, $g_t(\theta)\leq g_t(\theta')$ whenever
$0\leq \theta\leq \theta'$.

\item For every $t\geq 1$ and every pair of input histories
$x[1..t],x'[1..t]\in\mcal{U}^t$,
\[
\Pow\left(
\algM_t(R_t;x[1..t])
\|
\algM_t(R_t;x'[1..t])
\right)
=
g_t\left(
d\left(
\bigoplus_{\tau\in I_t}x[\tau],
\bigoplus_{\tau\in I_t}x'[\tau]
\right)
\right).
\]
\end{compactitem}
Here, $\algM_t(R_t;x[1..t])$ denotes the distribution of
$\algM_t(\omega_t;x[1..t])$ when $\omega_t$ is sampled from $R_t$.
\end{definition}

We can now state the certification theorem.  The only-if direction is
immediate, since non-adaptive adversaries are a special case of adaptive
adversaries.  The content is the converse: for quantitatively decomposable
mechanisms, adaptivity does not increase the worst-case power-function
profile.  Thus any finite-horizon non-adaptive $f$-DP proof automatically
lifts to the paired-simulation adaptive game.

\begin{theorem}[Adaptive Safety of Quantitatively Decomposable Mechanisms]
\label{thm:noise-ind-DP}
Consider the Hamming-style neighboring relation on equal-length input
sequences in $\mcal{U}$, where two sequences are neighboring if they differ
in at most one position.  Let $f$ be a power function.  Suppose
$\algM:\mcal{R}\times\mcal{U}^*\rightarrow\mcal{U}$ is quantitatively
decomposable in the sense of Definition~\ref{defn:quant_decomp}.  If
$\algM$ is non-adaptively $f$-DP in the paired-simulation sense of
Definition~\ref{defn:do}, for every finite horizon, then $\algM$ is
adaptively $f$-DP in the same sense, for every finite horizon.
\end{theorem}

\begin{proof}
We prove the following finite-horizon statement by induction on~$T$: for every
power function~$f$, if $\algM$ is non-adaptively $f$-DP up to horizon~$T$,
then $\algM$ is adaptively $f$-DP up to horizon~$T$.  This implies the theorem
because the theorem assumes the non-adaptive guarantee for every finite horizon.

We first reduce to deterministic adversaries.  A randomized adversary can be
viewed as first sampling its random seed and then running a deterministic
adversary.  For every fixed seed, the deterministic-adversary case gives the
desired power-function bound.  Taking the mixture over the same seed
distribution under the two secret bits preserves the bound by the joint
convexity property in Fact~\ref{fact:joint_convexity}.  Hence it suffices to
consider deterministic adaptive adversaries.

We also assume, without loss of generality, that the adversary observes the
entire output of $\algM$ in each round.  If the real adversary observes only a
view of the output, the real view is obtained from the full-output transcript
by post-processing, and the desired bound follows from data processing.  If
the mechanism has an initial output before the adversary sends any input, this
initial output has the same distribution under the two secret bits; it may be
included as an additional common first observation and does not affect the
privacy bound.  Thus we count rounds from the first adversarial input.

Let $\algA$ be a deterministic adaptive adversary in the paired simulation
$(\algM^\pair_0,\algM^\pair_1)$.  We pad terminated executions with a fixed
halt symbol, so all transcripts have length~$T$.  This padding is
post-processing and does not change the privacy claim.  We must prove
\[
\Pow\!\left(
\view(\algA \leftrightarrow \algM^\pair_0)
\|
\view(\algA \leftrightarrow \algM^\pair_1)
\right)
\leq f .
\]

The base case $T=1$ is immediate.  With only one input round, the adversary has
no previous mechanism output on which to adapt, so the adaptive and
non-adaptive paired simulations coincide.

Assume the statement holds for horizon $T-1$, and consider horizon $T>1$.
Let
\[
(x^{(0)}[1],x^{(1)}[1])
\]
be the first pair of inputs produced by $\algA$.  Since $\algA$ is
deterministic, this pair is fixed.  Let
\[
U^{(b)} := \algM_1(R_1;x^{(b)}[1])
\]
be the first-round output distribution under secret bit $b$.

We split into two cases.

\smallskip
\noindent\textbf{Case 1: the first round is the challenge round.}
Assume $x^{(0)}[1]\neq x^{(1)}[1]$, and write
\[
\Delta := d(x^{(0)}[1],x^{(1)}[1]).
\]
Because the paired simulation only allows Hamming-neighboring prefixes, once
the first inputs differ, every later input pair produced by the adversary must
have equal components.  That is, on every non-terminated transcript prefix and
for every $t\geq 2$, the adversary must output a pair of the form $(x_t,x_t)$.
We use the finite-horizon convention that terminated executions are padded
with a fixed halt symbol; the padded rounds have identical outputs under the
two secret bits and hence contribute the identity power function.

For each round~$t$, define
\[
h_t :=
\begin{cases}
g_t(\Delta), & \text{if } 1\in I_t,\\
\mathsf{Id}, & \text{if } 1\notin I_t.
\end{cases}
\]
We first show that, at every round~$t$ and every transcript prefix that can
arise before round~$t$, the conditional next-output distributions under the
two secret bits have power function exactly~$h_t$.

Fix such a transcript prefix.  Since the adversary is deterministic,
conditioning on this common prefix fixes all input pairs generated so far.
As argued above, the two input histories entering round~$t$ differ only in
the first coordinate.  If $1\notin I_t$, then the two aggregates
\[
\bigoplus_{\tau\in I_t}x^{(0)}[\tau]
\quad\text{and}\quad
\bigoplus_{\tau\in I_t}x^{(1)}[\tau]
\]
are identical.  Hence their distance is zero.  Applying
Definition~\ref{defn:quant_decomp} to identical histories gives
$g_t(0)=\mathsf{Id}$, and therefore the one-round power function is exactly
$\mathsf{Id}$.

If $1\in I_t$, then the two aggregates have the form
\[
x^{(0)}[1]\oplus c
\quad\text{and}\quad
x^{(1)}[1]\oplus c
\]
for the same common suffix aggregate~$c$ over the remaining indices in~$I_t$.
By the invariant-pair property,
\[
d(x^{(0)}[1]\oplus c,\;x^{(1)}[1]\oplus c)
=
d(x^{(0)}[1],x^{(1)}[1])
=
\Delta .
\]
Hence Definition~\ref{defn:quant_decomp} gives the exact one-round identity
\[
\Pow\!\left(
\algM_t(R_t;x^{(0)}[1..t])
\|
\algM_t(R_t;x^{(1)}[1..t])
\right)
=
g_t(\Delta)
=
h_t .
\]
Thus the one-round power function is exactly $h_t$, and $h_t$ depends only on
$t$, not on the transcript prefix.

We now justify that these exact one-round identities tensorize to an exact
identity for the whole adaptive transcript in this case.  Let $P_b$ be the
distribution of the full length-$T$ transcript under secret bit $b$.  For a
transcript prefix $s=y[1..t-1]$, let $K^b_{t,s}$ be the conditional
distribution of $y_t$ given the prefix~$s$ under bit~$b$.  By independent
decomposability, $K^b_{t,s}$ is generated using the fresh seed~$R_t$, which
is independent of the previous seeds.  From the preceding paragraph,
\[
\Pow(K^0_{t,s}\|K^1_{t,s})=h_t
\]
for every prefix~$s$ that can arise.

We use the following finite-experiment characterization of power functions.
For two finite distributions $P$ and $Q$, the power function
$\Pow(P\|Q)$ is determined by the decreasing rearrangement of the likelihood
ratio $Q(\omega)/P(\omega)$ under $\omega\sim P$; equivalently, two pairs of
finite distributions with the same power function have the same likelihood
ratio distribution under their first distribution, up to the usual splitting
of atoms.  Moreover, tensor product of power functions corresponds to taking
the product of independent likelihood ratios.

Apply this characterization to the conditional kernels.  For each prefix~$s$,
the likelihood ratio
\[
L_{t,s}(y_t):=\frac{K^1_{t,s}(y_t)}{K^0_{t,s}(y_t)}
\]
has, under $y_t\sim K^0_{t,s}$, the likelihood-ratio distribution determined
by~$h_t$.  This distribution is the same for every prefix~$s$.  Since
$R_t$ is fresh and independent of the previous seeds, the conditional
likelihood-ratio increment at round~$t$ is independent of the past after
conditioning on the prefix, and its distribution is determined only by~$h_t$.

For a full transcript $y[1..T]$, the likelihood ratio between $P_1$ and
$P_0$ factors as
\[
\frac{P_1(y[1..T])}{P_0(y[1..T])}
=
\prod_{t=1}^T
\frac{K^1_{t,y[1..t-1]}(y_t)}
     {K^0_{t,y[1..t-1]}(y_t)} .
\]
Under $P_0$, the factors in this product have the same joint distribution as
independent likelihood-ratio variables whose one-step power functions are
$h_1,\ldots,h_T$.  Therefore the power function of the full adaptive
transcript pair is exactly the tensor product:
\[
\Pow(P_0\|P_1)
=
\bigotimes_{t=1}^T h_t
=
\bigotimes_{\substack{t\in[1..T]\\ 1\in I_t}} g_t(\Delta),
\]
where the second equality uses $\mathsf{Id}\otimes q=q$ for every power
function~$q$.

Finally, consider any non-adaptive adversary that submits the fixed first
pair $x^{(0)}[1],x^{(1)}[1]$ and then submits identical inputs in every later
round.  The later common inputs may be chosen arbitrarily.  By the invariant
pair property, they only contribute common aggregation context and therefore
do not change the one-step power labels $h_t$.
 For this
non-adaptive input pair, the same tensor product is the power function of the
transcript pair:
\[
\bigotimes_{\substack{t\in[1..T]\\ 1\in I_t}} g_t(\Delta).
\]
Since $\algM$ is non-adaptively $f$-DP up to horizon~$T$, this tensor product
is at most~$f$.  Hence
\[
\Pow\!\left(
\view(\algA \leftrightarrow \algM^\pair_0)
\|
\view(\algA \leftrightarrow \algM^\pair_1)
\right)
\leq f
\]
in Case~1.

\smallskip
\noindent\textbf{Case 2: the first round is not the challenge round.}
Assume $x^{(0)}[1]=x^{(1)}[1]$.  Denote this common value by $c$.  Then
$U^{(0)}=U^{(1)}$ as distributions, and hence
\[
\Pow(U^{(0)}\|U^{(1)})=\mathsf{Id}.
\]

We now define a shifted mechanism for the remaining $T-1$ rounds after the
common first input~$c$ has been fixed.  Let $\algM^{(c)}$ be the mechanism
whose round-$s$ output, for $1\leq s\leq T-1$, is
\[
\algM^{(c)}_s(R_{s+1};z[1..s])
:=
\algM_{s+1}(R_{s+1};c,z[1..s]).
\]
This mechanism uses the independent seeds $R_2,\ldots,R_T$, so it is
independently decomposable.  It is also quantitatively decomposable.  Indeed,
for round~$s$ define
\[
I^{(c)}_s := \{\tau-1:\tau\in I_{s+1},\ \tau\geq 2\}\subseteq[1..s].
\]
If $1\notin I_{s+1}$, then the aggregate used by $\algM_{s+1}$ is exactly the
aggregate over $I^{(c)}_s$ in the suffix history.  If $1\in I_{s+1}$, then
the aggregate has the common left context~$c$, and the invariant-pair
property removes this common context.  Therefore, for every two suffix
histories $z[1..s]$ and $z'[1..s]$,
\[
\Pow\!\left(
\algM^{(c)}_s(R_{s+1};z[1..s])
\|
\algM^{(c)}_s(R_{s+1};z'[1..s])
\right)
=
g_{s+1}\!\left(
d\!\left(
\bigoplus_{\rho\in I^{(c)}_s}z[\rho],
\bigoplus_{\rho\in I^{(c)}_s}z'[\rho]
\right)
\right).
\]
Thus $\algM^{(c)}$ satisfies the same structural assumptions as $\algM$,
with the time indices shifted by one.

Let $\Lambda$ be the set of all Hamming-neighboring pairs of suffix histories
of length $T-1$.  For
\[
\lambda=(z^{(0)}[1..T-1],z^{(1)}[1..T-1])\in\Lambda,
\]
let $V_\lambda^{(b)}$ be the distribution of the length-$(T-1)$ transcript
generated by $\algM^{(c)}$ on the fixed suffix input $z^{(b)}[1..T-1]$.
Define
\[
g_\lambda := \Pow(V_\lambda^{(0)}\|V_\lambda^{(1)})
\quad\text{and}\quad
g := \sup_{\lambda\in\Lambda} g_\lambda .
\]
By definition of the supremum of power functions, $g$ is a power function.

We claim that $\algM^{(c)}$ is non-adaptively $g$-DP up to horizon $T-1$.
For deterministic non-adaptive suffix adversaries this is immediate from the
definition of~$g$.  Randomized non-adaptive suffix adversaries are mixtures of
deterministic ones, and the same bound follows from Fact~\ref{fact:joint_convexity}.

Moreover, $g\leq f$.  To see this, fix any $\lambda\in\Lambda$ and consider
the non-adaptive length-$T$ adversary for $\algM$ that submits the common first
input~$c$ in both branches and then submits the suffix pair~$\lambda$.  The
first-round output has the same distribution under the two secret bits, and
the remaining transcript has law $V_\lambda^{(b)}$ under bit~$b$.  Hence the
power function of this full non-adaptive transcript pair is
\[
\mathsf{Id}\otimes g_\lambda = g_\lambda .
\]
The non-adaptive $f$-DP assumption for $\algM$ gives $g_\lambda\leq f$ for
every $\lambda\in\Lambda$.  Taking the supremum over $\lambda$ gives
$g\leq f$.

Now condition on a possible first-round output $u$ in the support of
$U^{(0)}=U^{(1)}$.  After observing $u$, the original deterministic adaptive
adversary $\algA$ induces a deterministic adaptive suffix adversary
$\algA_u$ for the shifted mechanism $\algM^{(c)}$: it continues exactly as
$\algA$ would continue after the first transcript symbol~$u$.  Let
\[
W^{(b)}(u)
:=
\view(\algA_u \leftrightarrow (\algM^{(c)})^\pair_b)
\]
be the distribution of the remaining transcript under secret bit~$b$, and
write
\[
h_u := \Pow(W^{(0)}(u)\|W^{(1)}(u)).
\]
Since $\algM^{(c)}$ is quantitatively decomposable and non-adaptively
$g$-DP up to horizon $T-1$, the induction hypothesis gives
\[
h_u\leq g
\]
for every such first-round output~$u$.

Finally, the full transcript under secret bit~$b$ has the same distribution as
\[
(U,W^{(b)}(U)),
\]
where $U\sim U^{(0)}=U^{(1)}$.  The first coordinate has identical
distribution under the two secret bits, and, conditional on $U=u$, the
continuation distributions have power function at most~$g$.  Applying
Fact~\ref{fact:composition} with first-stage profile $\mathsf{Id}$ and
second-stage profile~$g$ yields
\[
\Pow\!\left(
\view(\algA \leftrightarrow \algM^\pair_0)
\|
\view(\algA \leftrightarrow \algM^\pair_1)
\right)
\leq
\mathsf{Id}\otimes g
=
g
\leq f .
\]
This proves Case~2.

The two cases complete the induction step.  Hence the finite-horizon adaptive
$f$-DP guarantee holds for every deterministic adaptive adversary.  As argued
at the beginning, joint convexity then extends the same guarantee to
randomized adaptive adversaries.  Therefore $\algM$ is adaptively $f$-DP for
every finite horizon.
\end{proof}

\begin{remark}[Vector-Valued Aggregates]
\label{rem:vector_quant_decomp}
Definition~\ref{defn:quant_decomp} states quantitative decomposability using a
scalar distance parameter.  This is sufficient for the theorem as used in the
main text.  Some standard vector-valued mechanisms, however, have an exact
one-round power profile that depends on the full difference vector rather
than only on its norm.  For example, for coordinate-wise Laplace noise on
$\mathbb{R}^k$, the exact profile of $u+\eta$ versus $u'+\eta$ may depend on
the vector $u-u'$.

The theorem and proof extend verbatim to this setting by replacing the scalar
distance $d(u,u')$ with an invariant difference map
$\Delta(u,u')$ taking values in some parameter space, such as
$\Delta(u,u')=u-u'\in\mathbb{R}^k$, and replacing $g_t(\theta)$ by a family
$g_t(v)$ indexed by this difference parameter.  The only property needed in
the proof is invariance under common aggregation context:
\[
\Delta(a\oplus u\oplus b,\;a\oplus u'\oplus b)=\Delta(u,u').
\]
Under this vector-indexed variant, the binary-tree prefix-sum mechanism with
coordinate-wise Laplace noise is covered exactly for vector updates in
$\mathbb{R}^k$.  In the main text, we state the simpler scalar version because
the usual $(\varepsilon,\delta)$ consequences can also be obtained by
upper-bounding the vector-indexed profile in terms of a norm, such as
$\|u-u'\|_1$ for coordinate-wise Laplace noise.
\end{remark}

\subsection{Why Independent Decomposability Alone Is Insufficient}
\label{subsec:ind_decomp_counterexample}

The certification theorem in Section~\ref{sec:ind_decomp} requires
quantitative decomposability, not merely independent decomposability.  The
following example shows why the stronger condition is needed.  Even when each
round uses fresh independent randomness and the round-$t$ output is a
deterministic function of the input prefix and the fresh seed, a
non-adaptive approximate-DP guarantee need not lift to adaptive adversaries.
The obstruction is that the one-round leakage can depend on the adaptive
common context, rather than only on the location and magnitude of the Hamming
discrepancy.

\paragraph{Counterexample with Fresh Independent Randomness.}
Consider a two-round mechanism with input and output alphabets $\{0,1\}$.
Write $x_1=b$ and $x_2=q$.  Using independent fresh randomness in the two
rounds, the mechanism outputs $Y_1,Y_2\in\{0,1\}$ with
\[
\Pr[Y_1=1\mid b=0]=0.68,\qquad
\Pr[Y_1=1\mid b=1]=0.65,
\]
and
\[
\begin{array}{c|cc}
& q=0 & q=1 \\ \hline
b=0 & \Pr[Y_2=1]=0.28 & \Pr[Y_2=1]=0.20 \\
b=1 & \Pr[Y_2=1]=0.19 & \Pr[Y_2=1]=0.28 .
\end{array}
\]
Since each output is a deterministic function of the current prefix and a
fresh independent seed, the mechanism is independently decomposable in the
sense of Definition~\ref{defn:indep_decomp}.

For every fixed input stream $(b,q)$, the transcript distribution is the
product distribution of $Y_1$ and $Y_2$.  Since two-sided $(0,\delta)$-DP is
equivalent to total variation distance at most $\delta$, it suffices to check
the four Hamming-neighboring pairs:
\[
\begin{array}{c|c}
\text{Fixed neighboring streams} & \text{Total variation distance} \\ \hline
(0,0)\text{ and }(1,0) & 0.09 \\
(0,1)\text{ and }(1,1) & 0.08 \\
(0,0)\text{ and }(0,1) & 0.08 \\
(1,0)\text{ and }(1,1) & 0.09 .
\end{array}
\]
Hence the mechanism is non-adaptively $(0,0.09)$-DP.

Now consider the deterministic adaptive adversary that chooses the challenge
pair $x^{(0)}_1=0$ and $x^{(1)}_1=1$ at round~$1$, and, after observing
$Y_1=y$, submits the common second input
$x^{(0)}_2=x^{(1)}_2=y$.  This is a valid adaptive Hamming-neighbor
adversary: the two submitted streams differ only at coordinate~$1$, but their
common second input depends on the first release.

The resulting transcript distributions are:
\[
\begin{array}{c|cc}
(Y_1,Y_2) & \text{World }0 & \text{World }1 \\ \hline
(0,0) & 0.2304 & 0.2835 \\
(0,1) & 0.0896 & 0.0665 \\
(1,0) & 0.5440 & 0.4680 \\
(1,1) & 0.1360 & 0.1820 .
\end{array}
\]
For $S=\{(0,1),(1,0)\}$, we have
\[
\Pr_0[S]-\Pr_1[S]=0.6336-0.5345=0.0991>0.09.
\]
Therefore,
\[
\operatorname{TV}(\mathsf{View}_0,\mathsf{View}_1)=0.0991,
\]
so the mechanism is not adaptively $(0,0.09)$-DP.

This does not contradict Theorem~\ref{thm:noise-ind-DP}, because the theorem
assumes quantitative decomposability.  The example shows that fresh
independent randomness alone is not enough: the adaptive adversary can choose
a common later input that changes how the earlier Hamming discrepancy is
revealed.  Quantitative decomposability rules out precisely this behavior by
forcing the one-round power function to depend only on the distance between
the relevant aggregates, with common aggregation context removed by the
invariant-pair property.

\section{Streaming Prefix-Sum Mechanism}
\label{sec:prefix-sum}

There have been numerous works on differentially private prefix-sum
algorithms since the \emph{binary tree mechanism} was
introduced~\cite{DBLP:conf/icalp/ChanSS10,DBLP:journals/tissec/ChanSS11}.
We first state the standard pure-DP primitive with a simultaneous
high-probability error bound, and then derive the adaptive,
always-error-bounded, and consistent variant needed by \randbin.

\noindent\textbf{Problem Setting.}
At each time step $t\geq 1$, an integer $x_t$ is passed to an
interactive mechanism, which releases an estimate $S_t$ of the prefix sum
\[
Q_t:=\sum_{\tau=1}^t x_\tau.
\]
The additive error at time~$t$ is $|S_t-Q_t|$, and the adversary may
choose later inputs after observing earlier releases.

\noindent\textbf{Neighboring Notion.}
Two equal-length input streams are neighboring if they differ in at most
one time step, where the two values at that step differ by exactly~$1$.
As aforementioned, an unbounded interactive mechanism is adaptively
$(\varepsilon,\delta)$-DP if its truncation to every finite horizon
$T>0$ is adaptively $(\varepsilon,\delta)$-DP.

As in~\cite{DBLP:journals/jacm/ChanCMS22}, we begin with a pure-DP
prefix-sum mechanism having a simultaneous high-probability error bound.

\begin{fact}[Pure-DP Prefix Sums~\cite{DBLP:journals/tissec/ChanSS11}]
\label{thm:dp_unbounded_prefix_sum}
For every $\varepsilon>0$, there exists an integer-valued interactive
mechanism $\hpresum_{\varepsilon}$, defined on streams in $\Z$, with the
following properties.

\begin{itemize}
\item \textbf{Non-adaptive privacy.}
For the unit-Hamming neighboring notion above,
$\hpresum_{\varepsilon}$ is non-adaptively
$(\varepsilon,0)$-DP. This guarantee holds for the stronger transcript
consisting of all noisy interval counts described below.

\item \textbf{Simultaneous error.}
For every $\beta\in(0,1)$, there is a deterministic nondecreasing function
\[
B_{\varepsilon,\beta}(t)
=
O\left(
\frac{1}{\varepsilon}\log(t+1)
\left(
\log(t+1)+\log\frac{1}{\beta}
\right)
\right)
\]
such that, with probability at least $1-\beta$,
\[
\left|
Y_t-\sum_{\tau=1}^t x_\tau
\right|
\leq B_{\varepsilon,\beta}(t)
\qquad\text{for every }t\geq1,
\]
where $Y_t$ is the estimate released at time~$t$. The simultaneous
success event may be chosen to depend only on the mechanism's noises.
Consequently, the same event applies when the inputs are selected
adaptively from previous releases.

\item \textbf{Certified core structure.}
The mechanism admits a $p$-sum implementation whose stronger core
transcript contains exactly one new noisy interval count in each original
round. The round-$t$ core output can be written as
\[
Z_t
=
\sum_{i\in I_t}x_i+\eta_t,
\qquad
I_t=[c_t..t]\subseteq[1..t].
\]
The schedule $(I_t)_{t\geq1}$ is data-independent, and the noises $\eta_t$
are fresh and independent across rounds. They may be taken to be symmetric
geometric noises. The exact one-round power function depends only on the
absolute integer shift between the two interval sums and is nondecreasing in
that shift. The published estimate $Y_t$ is a deterministic function of the
core releases available by time~$t$.

\end{itemize}
\end{fact}

The stronger-transcript statement follows from the standard $p$-sum view
of the hybrid mechanism. For completeness, its simultaneous
infinite-horizon error guarantee can be obtained as follows. For time~$t$,
apply the standard sub-exponential tail bound to the sum of the relevant
geometric noises with failure probability
\[
\beta_t:=\frac{\beta}{t(t+1)}.
\]
Since $\sum_{t\geq1}\beta_t=\beta$ and
\[
\log\frac{1}{\beta_t}
=
O\left(
\log(t+1)+\log\frac{1}{\beta}
\right),
\]
a union bound gives the stated error bound simultaneously for all
$t\geq1$. Moreover, after subtracting the true prefix sum, the error is a
deterministic function only of the core noises. Replacing the usual
Laplace perturbations by their symmetric geometric analogues preserves
pure DP and the same tail order while making all releases integer-valued.

\begin{theorem}[Adaptive Prefix Sums with Always-Valid Error and Consistency]
\label{th:prefix_sum}
For every $\varepsilon>0$, $\delta\in(0,1)$, and integer interval
$[a..b]\subseteq\Z$, there exists an interactive mechanism
$\presum_{\varepsilon,\delta}$, defined on all integer streams, that is
adaptively $(\varepsilon,\delta)$-DP for unit-Hamming neighboring streams.

Moreover, there is a deterministic nondecreasing integer-valued function
\[
E_{\varepsilon,\delta}:\Z_{\geq0}\rightarrow\Z_{\geq0},
\qquad
E_{\varepsilon,\delta}(0)=0,
\]
such that, for every $t\geq1$,
\[
E_{\varepsilon,\delta}(t)
=
O\left(
1+
\frac{1}{\varepsilon}\log(t+1)
\left(
\log(t+1)+\log\frac{1}{\delta}
\right)
\right).
\]
On every execution whose inputs satisfy $x_t\in[a..b]$, the integer outputs
$S_t$ satisfy, simultaneously for every $t\geq1$,
\begin{align*}
\left|S_t-\sum_{\tau=1}^t x_\tau\right|
&\leq E_{\varepsilon,\delta}(t),
&&\text{\emph{(always-valid error)}}\\
S_t-S_{t-1}
&\in[a..b],
&&\text{\emph{(consistency)}}
\end{align*}
where $S_0:=0$. These properties hold for every realization of the
mechanism's randomness.
\end{theorem}

\begin{proof}
Let $\hpresum_{\varepsilon}$ be the mechanism in
Fact~\ref{thm:dp_unbounded_prefix_sum}, and abbreviate
\[
B(t):=B_{\varepsilon,\delta}(t).
\]
Define
\[
E_{\varepsilon,\delta}(0):=0,
\qquad
E_{\varepsilon,\delta}(t):=\left\lceil B(t)\right\rceil+2
\quad(t\geq1).
\]
Because $B$ is nondecreasing, $E_{\varepsilon,\delta}$ is a nondecreasing
integer-valued function with the asymptotic bound stated in the theorem.
We construct the claimed mechanism in three stages.

\smallskip
\noindent\emph{1. Adaptive privacy of the raw mechanism.}
Use the stronger core transcript from
Fact~\ref{thm:dp_unbounded_prefix_sum}. Its round-$t$ release is
\[
Z_t
=
\sum_{\tau\in I_t}x_\tau+\eta_t,
\]
where $I_t\subseteq[1..t]$ is data-independent and $\eta_t$ is fresh and
independent of all other rounds.

Take the aggregation operation to be integer addition and define
\[
d(u,u'):=|u-u'|.
\]
Then $(d,+)$ is an invariant pair because, for all integers
$a,b,u,u'$,
\[
d(a+u+b,\;a+u'+b)=|u-u'|=d(u,u').
\]

For every nonnegative integer $k$, define
\[
\widetilde g_t(k)
:=
\Pow(\eta_t\|\eta_t+k).
\]
By symmetry, the same power function is obtained for a shift by $-k$.
Moreover, for symmetric geometric noise,
$\widetilde g_t(k)$ is nondecreasing in~$k$. Extend it to
$\theta\geq0$ by
\[
g_t(\theta):=\widetilde g_t(\lfloor\theta\rfloor).
\]
Because every relevant distance is an integer, this extension preserves
the required exact identities.

Translation invariance and symmetry now give, for every pair of integer
input histories,
\[
\Pow\left(
\sum_{\tau\in I_t}x_\tau+\eta_t
\;\middle\|\;
\sum_{\tau\in I_t}x'_\tau+\eta_t
\right)
=
g_t\left(
\left|
\sum_{\tau\in I_t}x_\tau-
\sum_{\tau\in I_t}x'_\tau
\right|
\right).
\]
Thus the noisy-count core is quantitatively decomposable with respect to
the invariant pair $(d,+)$ in the sense of
Definition~\ref{defn:quant_decomp}. By
Fact~\ref{thm:dp_unbounded_prefix_sum}, this core is non-adaptively
$(\varepsilon,0)$-DP for every finite horizon. Therefore,
Theorem~\ref{thm:noise-ind-DP} implies that it is adaptively
$(\varepsilon,0)$-DP for every finite horizon. Since the raw prefix
estimates are causal post-processing of the core transcript,
$\hpresum_{\varepsilon}$ is also adaptively
$(\varepsilon,0)$-DP.

\smallskip
\noindent\emph{2. An always-valid error bound.}
Given the raw candidate $Y_t$ at time~$t$, define the clipped release
\[
\widehat Y_t
:=
\operatorname{clip}\left(
Y_t,\,
\left[
Q_t-E_{\varepsilon,\delta}(t),\,
Q_t+E_{\varepsilon,\delta}(t)
\right]
\right),
\]
where
\[
\operatorname{clip}(y,[\ell,u])
:=
\min\{\max\{y,\ell\},u\}.
\]
This immediately gives
$|\widehat Y_t-Q_t|\leq E_{\varepsilon,\delta}(t)$ on every execution.
Because the clipping interval depends on the private prefix sum, however,
adaptive privacy does not follow from ordinary post-processing. We prove it
directly.

Fix a finite horizon~$T$ and, by joint convexity, a deterministic adaptive
adversary~$\algA$. For every candidate released transcript
$v[1..T]$, let
\[
x_t^{(c)}(v[1..t-1]),
\qquad c\in\{0,1\},
\]
be the $c$-th input selected by $\algA$ at time~$t$ after observing
$v[1..t-1]$, and define the counterfactual prefix sums
\[
Q_t^{(c)}(v[1..t-1])
:=
\sum_{\tau=1}^t
x_\tau^{(c)}(v[1..\tau-1]).
\]
We may assume that the adversary generates unit-Hamming neighboring
streams after every transcript prefix; its behavior after a prefix at
which the paired simulation would terminate is immaterial and can be
redefined arbitrarily. Hence
\begin{equation}
\label{eq:counterfactual_prefix_close}
\left|
Q_t^{(0)}(v[1..t-1])
-
Q_t^{(1)}(v[1..t-1])
\right|
\leq1
\end{equation}
for every transcript prefix and every $t\leq T$.

Let $P_c$ and $\widehat P_c$ be the laws of the raw and clipped released
transcripts, respectively, in branch~$c$. Define the common measurable set
\begin{equation}
\label{eq:prefix_sum_common_safe_set}
G_T
:=
\left\{
v[1..T]:
\left|
v_t-Q_t^{(c)}(v[1..t-1])
\right|
<
E_{\varepsilon,\delta}(t)
\text{ for every $t\leq T$ and $c\in\{0,1\}$}
\right\}.
\end{equation}
This is one fixed subset of the transcript space, defined using both
counterfactual input sequences.

For branch~$c$, let
\[
H_c
:=
\left\{
\left|
Y_t-Q_t^{(c)}(Y[1..t-1])
\right|
\leq B(t)
\text{ for every }t\leq T
\right\}.
\]
The noise-only simultaneous error event in
Fact~\ref{thm:dp_unbounded_prefix_sum} gives
$P_c(H_c)\geq1-\delta$, even for adaptively selected inputs. On $H_c$,
\eqref{eq:counterfactual_prefix_close} implies, for every $t\leq T$,
\begin{align*}
\left|
Y_t-Q_t^{(c)}(Y[1..t-1])
\right|
&\leq B(t)<E_{\varepsilon,\delta}(t),\\
\left|
Y_t-Q_t^{(1-c)}(Y[1..t-1])
\right|
&\leq B(t)+1<E_{\varepsilon,\delta}(t).
\end{align*}
Hence $H_c\subseteq G_T$ and
\begin{equation}
\label{eq:raw_prefix_sum_bad_set}
P_c(G_T^c)\leq\delta.
\end{equation}

Couple the raw and clipped executions in a fixed branch~$c$ using the same
complete noise sequence. We claim that
\begin{equation}
\label{eq:raw_clipped_prefix_sum_identity}
Y[1..T]\in G_T
\quad\Longleftrightarrow\quad
\widehat Y[1..T]\in G_T,
\end{equation}
and that the two transcripts are identical whenever either condition
holds. For the forward direction, induct on~$t$. If
$Y[1..t-1]=\widehat Y[1..t-1]$, then the adversary supplies the same input
and the shared noises give the same raw candidate in both executions.
Because $Y[1..T]\in G_T$, that candidate lies strictly inside the
branch-$c$ clipping interval, so clipping leaves it unchanged. For the
reverse direction, the same induction applies: if
$\widehat Y[1..T]\in G_T$, then $\widehat Y_t$ lies strictly inside the
clipping interval; if clipping had changed the raw candidate, the released
value would instead equal an endpoint. Thus clipping did not occur.

It follows from \eqref{eq:raw_clipped_prefix_sum_identity} that, for every
measurable transcript event~$A$,
\begin{equation}
\label{eq:restricted_prefix_sum_laws}
\widehat P_c(A\cap G_T)=P_c(A\cap G_T),
\qquad
\widehat P_c(G_T^c)=P_c(G_T^c)\leq\delta.
\end{equation}
Using adaptive pure DP of the raw mechanism, we obtain
\begin{align*}
\widehat P_0(A)
&\leq
\widehat P_0(A\cap G_T)+\widehat P_0(G_T^c)\\
&=
P_0(A\cap G_T)+P_0(G_T^c)\\
&\leq
e^\varepsilon P_1(A\cap G_T)+\delta\\
&=
e^\varepsilon\widehat P_1(A\cap G_T)+\delta\\
&\leq
e^\varepsilon\widehat P_1(A)+\delta.
\end{align*}
The same argument with the branches interchanged gives the reverse
inequality. Thus the clipped mechanism is adaptively
$(\varepsilon,\delta)$-DP for every finite horizon~$T$, and hence for the
unbounded interaction.

\smallskip
\noindent\emph{3. Consistency.}
Apply the following deterministic online projection to the clipped
transcript:
\[
S_0:=0,
\qquad
S_t
:=
\operatorname{clip}\left(
\widehat Y_t,\,
[S_{t-1}+a,\;S_{t-1}+b]
\right).
\]
This is causal post-processing, so it preserves adaptive
$(\varepsilon,\delta)$-DP: any adversary interacting with the projected
releases, together with this causal wrapper, induces an adversary for the
clipped mechanism. The projection immediately gives
$S_t-S_{t-1}\in[a..b]$.

It remains to verify the error. Suppose inductively that
$|S_{t-1}-Q_{t-1}|\leq E_{\varepsilon,\delta}(t-1)$. Since
$x_t\in[a..b]$ and $E_{\varepsilon,\delta}$ is nondecreasing,
\[
\operatorname{dist}\left(
Q_t,\,
[S_{t-1}+a,S_{t-1}+b]
\right)
\leq
|Q_{t-1}-S_{t-1}|
\leq
E_{\varepsilon,\delta}(t).
\]
Projection onto an interval satisfies
\[
\left|
\operatorname{clip}(y,J)-q
\right|
\leq
\max\{|y-q|,\operatorname{dist}(q,J)\}.
\]
Together with
$|\widehat Y_t-Q_t|\leq E_{\varepsilon,\delta}(t)$, this yields
$|S_t-Q_t|\leq E_{\varepsilon,\delta}(t)$. The induction starts from
$S_0=Q_0=0$, proving both claimed execution-wise properties.
\end{proof}

\subsection{Prefix Sums via Lower-Triangular Matrix Factorizations}
\label{sec:prefixsum_factor}

The mechanism above is the integer prefix-sum scheduler used internally by
\randbin. We now separately analyze the Gaussian lower-triangular
aggregators.
These aggregators require adaptive
privacy, but not the probability-one error or consistency properties of
Theorem~\ref{th:prefix_sum}.

Fix a horizon $T$ and write a scalar stream as $x\in\R^T$. The prefix sums
form the linear workload
\[
M_{\mathsf{count}}x\in\R^T,
\qquad
M_{\mathsf{count}}[t,i]=\mathbf{1}[t\geq i],
\]
so $(M_{\mathsf{count}}x)_t=\sum_{i=1}^t x_i$.

We use event-level neighbors $x\sim x'$ satisfying
\[
x-x'=\xi e_i
\qquad\text{for some }i\in[T]\text{ and }|\xi|\leq\Delta.
\]
For scalar unit-count streams, $\Delta=1$. For vector-valued updates, the
two streams differ in at most one time block, the $\ell_2$-norm of that
block difference is at most~$\Delta$, and the corresponding linear
operator is $R\otimes I$.

\noindent\emph{One-hop versus two-hop accounting.}
Theorem~\ref{th:prefix_sum} is used for \randbin's internal scheduler under
the unit-Hamming notion stated above. The privacy bounds in this subsection
are likewise primitive-level, distance-one guarantees. When a Gaussian
prefix-sum mechanism is used as the downstream aggregator, however, a single
input edit may induce at most two adjacent Hamming changes at \randbin's
emitted-update interface (Lemma~\ref{lem:randbin_wrapper}). Hence adaptive
group privacy (Fact~\ref{fact:multi_hop_DP}) gives
\[
\left(
2\widehat\varepsilon_a,\,
(1+e^{\widehat\varepsilon_a})\widehat\delta_a
\right)\text{-DP}.
\]
We therefore set
\[
\widehat\varepsilon_a=\frac{\varepsilon_a}{2},
\qquad
\widehat\delta_a=
\frac{\delta_a}{1+e^{\varepsilon_a/2}},
\]
so that the total downstream cost is at most
$(\varepsilon_a,\delta_a)$.

\paragraph{Privacy Notions: R\'enyi DP (RDP) and zCDP.}
We first state the usual fixed-stream, non-adaptive guarantees; their
adaptive lift is established below. Let $P,Q$ have densities $p,q$ with
respect to a common base measure. For $\alpha>1$,
\[
D_\alpha(P\|Q)
:=
\frac{1}{\alpha-1}
\log\int p(y)^\alpha q(y)^{1-\alpha}\,dy.
\]
A fixed-stream mechanism $\mcal{M}$ is
$(\alpha,\varepsilon_\alpha)$-RDP if, for all neighboring $x\sim x'$,
\[
D_\alpha\bigl(\mcal{M}(x)\|\mcal{M}(x')\bigr)
\leq\varepsilon_\alpha.
\]
It is $\rho$-zCDP if, for every $\alpha>1$ and every neighboring
$x\sim x'$,
\[
D_\alpha\bigl(\mcal{M}(x)\|\mcal{M}(x')\bigr)
\leq\rho\alpha.
\]
Equivalently, $\rho$-zCDP implies
$(\alpha,\rho\alpha)$-RDP simultaneously for every $\alpha>1$.

\paragraph{Converting RDP to $(\varepsilon,\delta)$-DP.}
If $\mcal{M}$ satisfies $(\alpha,\varepsilon_\alpha)$-RDP for some
$\alpha>1$, then, for every $\delta\in(0,1)$, it satisfies
$(\varepsilon,\delta)$-DP with
\[
\varepsilon
=
\varepsilon_\alpha+\frac{\log(1/\delta)}{\alpha-1}.
\]
Optimizing over the order gives
\[
\varepsilon(\delta)
=
\inf_{\alpha>1}
\left\{
\varepsilon_\alpha+\frac{\log(1/\delta)}{\alpha-1}
\right\}.
\]

\paragraph{Lower-Triangular Factorization Mechanism.}
Let $M_{\mathsf{count}}=LR$ be a possibly rectangular factorization, where
$R\in\R^{m\times T}$ and $L\in\R^{T\times m}$ are lower-triangular in
time, so all core releases and published outputs are causal. Define the
column sensitivity
\[
s(R):=\max_{i\in[T]}\|Re_i\|_2=\|R\|_{1\to2}.
\]
Consider
\[
\mcal{M}_{L,R,\sigma}(x)
:=
L(Rx+z),
\qquad
z\sim\mcal{N}(0,\sigma^2I_m),
\]
where the coordinates of $z$ are independent. Since multiplication by
$L$ is post-processing, privacy is governed by the Gaussian core
$Rx+z$.

\begin{lemma}[zCDP/RDP Bounds for $L(Rx+z)$]
\label{lem:factor_privacy_constants}
For neighboring $x\sim x'$ differing at coordinate~$i$,
\[
\|R(x-x')\|_2
=
|\xi|\,\|Re_i\|_2
\leq
\Delta s(R).
\]
Hence, in the fixed-stream experiment,
$\mcal{M}_{L,R,\sigma}$ satisfies $\rho$-zCDP with
\[
\rho=\frac{\Delta^2s(R)^2}{2\sigma^2},
\]
and, for every $\alpha>1$, it satisfies
$(\alpha,\varepsilon_\alpha)$-RDP with
\[
\varepsilon_\alpha
=
\alpha\rho
=
\frac{\alpha\Delta^2s(R)^2}{2\sigma^2}.
\]
\end{lemma}

\paragraph{Binary Tree Mechanism as a Factorization.}
Assume first that $T=2^h$ for an integer $h\geq0$; otherwise, pad to the
next power of two. Index the complete binary-tree nodes by their dyadic
intervals. Let $q=R_{\mathsf{tree}}x$ be the vector of all dyadic-interval
sums, where $R_{\mathsf{tree}}$ is the corresponding $0/1$ incidence
matrix. Each leaf belongs to exactly its $h+1$ ancestors, and each prefix
$[1..t]$ has a canonical dyadic partition of size at most $h+1$.
Therefore, there is a causal reconstruction matrix $L_{\mathsf{tree}}$
such that
$M_{\mathsf{count}}=L_{\mathsf{tree}}R_{\mathsf{tree}}$.

\begin{proposition}[Privacy Bounds for the Binary-Tree Factorization]
\label{prop:tree_zcdp_rdp}
For $T=2^h$,
\[
s(R_{\mathsf{tree}})^2
=
\|R_{\mathsf{tree}}\|_{1\to2}^2
=
h+1.
\]
Consequently, the Gaussian tree mechanism
$L_{\mathsf{tree}}(R_{\mathsf{tree}}x+z)$, with
$z\sim\mcal{N}(0,\sigma^2I_m)$, is non-adaptively
$\rho_{\mathsf{tree}}$-zCDP with
\[
\rho_{\mathsf{tree}}
=
\frac{\Delta^2(h+1)}{2\sigma^2},
\]
and, for every $\alpha>1$, is
$(\alpha,\varepsilon_{\alpha,\mathsf{tree}})$-RDP with
\[
\varepsilon_{\alpha,\mathsf{tree}}
=
\frac{\alpha\Delta^2(h+1)}{2\sigma^2}.
\]
For general $T$, the same bounds hold with
$h:=\lceil\log_2T\rceil$.
\end{proposition}

\paragraph{A Smooth Lower-Triangular Factorization~\cite{DBLP:conf/icml/FichtenbergerHU23}.}
Fichtenberger et al.~give an explicit Toeplitz lower-triangular
factorization $M_{\mathsf{count}}=LR$ with $L=R$, defined by
\[
f(0)=1,
\qquad
f(k)=\left(\frac{2k-1}{2k}\right)f(k-1)
\quad(k\geq1),
\qquad
L[t,i]=R[t,i]=f(t-i)
\quad(t\geq i),
\]
and zero otherwise. Let $\gamma$ be the Euler--Mascheroni constant and
define
\[
\Psi(T)
:=
1-\frac{1-\gamma}{\pi}+\frac{\ln T}{\pi}+\frac{2}{T}.
\]

\begin{proposition}[Privacy Bounds for the Fichtenberger--Henzinger--Upadhyay Factorization]
\label{prop:fhu_zcdp_rdp}
For the factorization above,
\[
s(R)^2
=
\|R\|_{1\to2}^2
=
\|L\|_{2\to\infty}^2
\leq
\Psi(T).
\]
Consequently, $L(Rx+z)$ with
$z\sim\mcal{N}(0,\sigma^2I_T)$ is non-adaptively
$\rho_{\mathsf{FHU}}$-zCDP with
\[
\rho_{\mathsf{FHU}}
=
\frac{\Delta^2\|R\|_{1\to2}^2}{2\sigma^2}
\leq
\frac{\Delta^2\Psi(T)}{2\sigma^2},
\]
and, for every $\alpha>1$, is
$(\alpha,\varepsilon_{\alpha,\mathsf{FHU}})$-RDP with
\[
\varepsilon_{\alpha,\mathsf{FHU}}
=
\frac{\alpha\Delta^2\|R\|_{1\to2}^2}{2\sigma^2}
\leq
\frac{\alpha\Delta^2\Psi(T)}{2\sigma^2}.
\]
\end{proposition}

\paragraph{Adaptive-Safety Audit.}
For each row $r_j^\top$ of $R$, let $\tau(j)$ be its public release time,
so that
\[
\supp(r_j)\subseteq[1..\tau(j)].
\]
Let
\[
J_t:=\{j:\tau(j)=t\},
\]
and let $A_t$ be the matrix whose rows are
$\{r_j^\top:j\in J_t\}$. Package all rows released at the same original
input time into the single vector-valued core release
\[
Z_t=A_tx+\eta_t,
\qquad
\eta_t\sim\mcal{N}(0,\sigma^2I_{|J_t|}).
\]
If $J_t=\emptyset$, take $Z_t$ to be a fixed data-independent dummy
release. The noise vectors $\eta_t$ are fresh and independent across
original input times. Thus there is exactly one core output per round, as
required by the interactive model.

To certify quantitative decomposability, represent the scalar update at
coordinate~$i$ by the vector $x_ie_i\in\R^T$ and aggregate these vectors
using addition. For vector-valued streams, use the analogous block vector
$e_i\otimes x_i$. 

In Definition~\ref{defn:quant_decomp}, take $I_t=[1..t]$. Because every
row of $A_t$ is supported on $[1..t]$,
\[
A_t\left(\sum_{i\in I_t}x_ie_i\right)=A_tx.
\]
For vector-valued streams, the same identity holds with $A_t\otimes I$.

Define the invariant difference map
\[
D(u,u'):=u'-u.
\]
It satisfies
\[
D(a+u+b,\;a+u'+b)=u'-u=D(u,u').
\]
For a discrepancy vector~$v$, define
\[
g_t(v)
:=
\Pow\left(
\eta_t
\;\middle\|\;
\eta_t+A_tv
\right).
\]
Translation invariance gives the exact identity
\[
\Pow\left(
A_tx+\eta_t
\;\middle\|\;
A_tx'+\eta_t
\right)
=
g_t(x'-x).
\]
Thus the Gaussian core is quantitatively decomposable under the
vector-indexed extension in
Remark~\ref{rem:vector_quant_decomp}. Since the rows released by time~$t$
depend only on $x[1..t]$, the core is causal, and its noise vectors are
fresh and independent.

It remains to connect the lifted power-function guarantee to the stated
RDP and zCDP bounds. Define
\[
\mu:=\frac{\Delta s(R)}{\sigma},
\]
and let $f_\mu$ denote the power function of the one-dimensional Gaussian
shift experiment
\[
\mcal{N}(0,1)
\quad\text{versus}\quad
\mcal{N}(\mu,1).
\]
For every fixed pair of neighboring streams,
\[
\frac{\|R(x-x')\|_2}{\sigma}
\leq\mu.
\]
Consequently, the non-adaptive Gaussian-core experiment has power function
at most~$f_\mu$. Theorem~\ref{thm:noise-ind-DP} therefore gives the same
$f_\mu$ upper bound for the transcript generated against every adaptive
adversary.

If $V_0,V_1$ denote the resulting adaptive transcript laws, Gaussian
experiment domination and data processing for R\'enyi divergence give,
for every $\alpha>1$,
\[
D_\alpha(V_0\|V_1)
\leq
\frac{\alpha\mu^2}{2}
=
\frac{\alpha\Delta^2s(R)^2}{2\sigma^2}.
\]
The same bound holds with the two branches interchanged. Hence the
adaptive core satisfies the same
\[
\rho=\frac{\Delta^2s(R)^2}{2\sigma^2}
\]
zCDP bound and the same
$(\alpha,\alpha\rho)$-RDP bounds as in the fixed-stream analysis. Since
$L$ is lower triangular, every published output is causal
post-processing of the core releases available at that time. The same
adaptive RDP/zCDP bounds, and their converted
$(\varepsilon,\delta)$-DP guarantees, therefore hold for the published
mechanism.

\paragraph{Takeaway.}
Both Gaussian mechanisms fit the template $L(Rx+z)$, but have different
column-sensitivity bounds: the binary tree has
$s(R)^2=\Theta(\log_2T)$, while the smooth Toeplitz factorization satisfies
\[
s(R)^2
\leq
1-\frac{1-\gamma}{\pi}+\frac{\ln T}{\pi}+\frac{2}{T}.
\]
The adaptive-safety audit shows that these same privacy bounds hold under
feedback-adaptive input selection.

\section{Other Related Work}
\label{sec:other_related}

As the most relevant related works have already been covered in the introduction,
we will elaborate further on some related aspects.

\noindent \textbf{Choices of Privacy Measure.}
The central idea in differential privacy~\cite{DBLP:conf/icalp/Dwork06} is that
if $V_0$ and $V_1$ are two distributions of outputs produced by a mechanism
from two neighboring inputs, then those two distributions should be close.
The classical $(\varepsilon, \delta)$-DP notion uses two parameters to quantify closeness, where a smaller value in each parameter means that the two distributions are closer.  However, since distributions are inherently complex
objects, some information on the two distributions will be inevitably lost
when they are compared using just two parameters.  When the same Gaussian noise
is used to mask two different vectors, researchers
have discovered that the \renyi divergence~\cite{Miro17}
can capture the variance of the Gaussian distribution perfectly, and
hence, can quantify the closeness of two such distributions (with the same
variance but different means).

In general, any useful way to quantify privacy guarantees must satisfy
the property that if the output satisfies certain privacy requirement,
then any further processing of the output cannot
violate that specific requirement.  This can be formally formulated by
requiring that the divergence --  used for measuring how different two distributions
are -- must satisfy the \emph{data processing inequality};
it is worth noting that some common distance notion such as the $\ell_2$-norm
do not satisfy this property.

Instead of just using a few parameters to capture the closeness
of two distributions, \emph{tradeoff functions}~\cite{dong2022gaussian}
have been proposed to define differential privacy, because a tradeoff
function can capture all the essential information
about how two distributions differ in the sense that any divergence
satisfying the data processing inequality can be recovered from
the tradeoff function.  Indeed, tradeoff functions offer a powerful tool
to describe the composition of private mechanisms.

However, one notational inconvenience is that a larger tradeoff (measured by pointwise comparison) means that the two distributions are closer,
which has the opposite interpretation from other divergence parameters
such as $\varepsilon$ and $\delta$.  In fact,
an \emph{ad hoc} concept of \emph{generalized probability distance}
has been defined in~\cite{DBLP:conf/stoc/Vadhan023} to reverse
the direction of the inequalities such that it will be consistent with the notion of distance.
On the other hand, a simpler way to achieve this notation consistency
is to replace a tradeoff function with its complement that
is known as a \emph{power function}, which naturally preserves all the equivalent
mathematical properties. As we shall see
in Definition~\ref{defn:power}, a power function also has an intuitive
description using the fractional knapsack problem.

As illustrated in~\cite{DBLP:conf/stoc/Vadhan023,DBLP:conf/innovations/ZhouZCS24}, if one uses such a powerful tool to define any new notion of differential privacy, then a single composition theorem (such as our Theorem~\ref{thm:DP-comp}) will be sufficient to recover any composition result from the classical notion to the new notion of privacy.  Therefore,
it would not be necessary to reconstruct individual advanced composition
theorems~\cite{DBLP:conf/focs/DworkRV10,DBLP:conf/icml/KairouzOV15}.

\noindent \textbf{Hamming vs Edit Neighboring Notions.}
In~\cite{DBLP:conf/nips/BirrellEBP24},
two neighboring notions are considered for \textbf{static}
databases:
\begin{compactitem}
\item \emph{Hamming-style.} Two databases have the same number of elements,
and they differ in at most one element.

\item \emph{Edit-style.}  One element from one database is deleted
to form the other database.
\end{compactitem}

They considered fixed-size mini-batches, which may be sampled in two ways: with or without replacement.  Since they considered sampling from static databases,
the difference between Hamming- vs edit-style neighboring static databases
would not have such a stark contrast as streams,
as  deleting the first element of a stream can cause it to change in every position.

\noindent \textbf{Other Recent Works on Concurrent Composition.}
A more general notion of concurrent composition is considered in~\cite{DBLP:conf/ccs/HaneySTVVX023}, where the privacy parameters of mechanisms can be adaptively chosen.
However, as in~\cite{DBLP:conf/stoc/Vadhan023}, each interactive mechanism is associated with a single static database on which its neighboring relation is defined; the adversary may interact adaptively with the mechanism, but neighboring inputs differ only in this underlying database, not in a dynamic stream of updates.

Concurrent composition for mechanisms with adaptively chosen privacy
parameters is also considered in~\cite{Henzinger2026Continual}, but for neighboring dynamic databases.
They also give a formulation based on adaptive DP~\cite{DBLP:conf/nips/DenisovMRST22}, expressed via a left-or-right style distinguishing game with a verification function, which is essentially the same as our \emph{paired simulation} in Definition~\ref{def:paired_sim}.

Note that both works consider composition where each mechanism has a single neighboring relation on its dataset (static or dynamic). In order to capture modular composition in which neighboring notions are defined separately on both the input and the output of a mechanism---and to reason about neighbor-preserving transformations that change the neighboring structure---we need a more refined notion of \emph{neighbor-preserving} paired simulation, given in Definition~\ref{def:npp}.

\end{document}